\documentclass[a4paper,oneside,english]{amsart}
\pdfoutput=1
\usepackage{lmodern}
\usepackage[T1]{fontenc}
\usepackage[latin9]{inputenc}
\usepackage{amsthm}
\usepackage{amstext}
\usepackage{amssymb}
\usepackage{graphicx}

\makeatletter

\pdfpageheight\paperheight
\pdfpagewidth\paperwidth

\numberwithin{equation}{section}
\numberwithin{figure}{section}
\theoremstyle{plain}
\newtheorem{thm}{\protect\theoremname}
  \theoremstyle{definition}
  \newtheorem{defn}[thm]{\protect\definitionname}
  \theoremstyle{remark}
  \newtheorem*{rem*}{\protect\remarkname}
  \theoremstyle{definition}
  \newtheorem{example}[thm]{\protect\examplename}
  \theoremstyle{plain}
  \newtheorem{lem}[thm]{\protect\lemmaname}

\makeatother

\usepackage{babel}
  \providecommand{\definitionname}{Definition}
  \providecommand{\examplename}{Example}
  \providecommand{\lemmaname}{Lemma}
  \providecommand{\remarkname}{Remark}
\providecommand{\theoremname}{Theorem}

\begin{document}
\begin{center}
{\Large ENGG4920CP Thesis II}
\par\end{center}{\Large \par}

\begin{center}
{\Large \vspace{1.5cm}
}
\par\end{center}{\Large \par}

\begin{center}
{\Huge O}{\huge N }{\Huge M}{\huge ULTI-RATE }{\Huge S}{\huge EQUENTIAL
}{\Huge D}{\huge ATA }{\Huge T}{\huge RANSMISSION}
\par\end{center}{\huge \par}

\begin{center}
{\Large \vspace{1.5cm}
}
\par\end{center}{\Large \par}

\begin{center}
{\Large BY}
\par\end{center}{\Large \par}

\begin{center}
{\Large \vspace{0.5cm}
}
\par\end{center}{\Large \par}

\begin{center}
{\LARGE Cheuk Ting LI}
\par\end{center}{\LARGE \par}

\begin{center}
{\Large \vspace{1.5cm}
}
\par\end{center}{\Large \par}

\begin{center}
{\Large A FINAL YEAR PROJECT REPORT SUBMITTED IN PARTIAL FULFILLMENT
OF THE REQUIREMENTS FOR THE DEGREE OF BACHELOR OF INFORMATION ENGINEERING
DEPARTMENT OF INFORMATION ENGINEERING THE CHINESE UNIVERSITY OF HONG
KONG}
\par\end{center}{\Large \par}

\begin{center}
{\Large \vspace{0.5cm}
}
\par\end{center}{\Large \par}

\begin{center}
{\Large May, 2012}
\par\end{center}{\Large \par}

\begin{center}
\vspace{5cm}

\par\end{center}

\pagebreak

\title{On Multi-rate Sequential Data Transmission}

\author{Cheuk Ting Li}
\begin{abstract}
In this report, we investigate the data transmission model in which
a sequence of data is broadcasted to a number of receivers. The receivers,
which have different channel capacities, wish to decode the data sequentially
at different rates. Our results are applicable to a wide range of
scenarios. For instance, it can be employed in the broadcast streaming
of a video clip through the internet, so that receivers with different
bandwidths can play the video at different speed. Receivers with greater
bandwidths can provide a smooth playback, while receivers with smaller
bandwidths can play the video at a slower speed, or with short pauses
or rebuffering.
\end{abstract}
\maketitle

\section{Introduction\label{sec:Introduction}}

Consider the scenario in which a long video clip has to be transmitted
to a number of receivers having different packet loss ratios. One
approach is to divide the video data into blocks of $K$ packets,
encode each block into $L\ge K$ encoding packets, and then transmit
the blocks to the receiver sequentially. Using random linear projections
or any capacity-achieving erasure code, the receiver can decode the
block if about $K$ out of $L$ packets are received. This method,
which we call a blockwise code, can only cater for the need of the
receiver with packet loss probability less than $1-K/L$.

To suit the need of different receivers, we can perform time multiplexing
on two blockwise codes at different rates. Cosider Blockwise code
1 and Blockwise code 2, which use random linear projections to encode
each block of $K$ packets into $L_{1}$ and $L_{2}$ packets respectively
($L_{1}<L_{2}$). Denote the $i$-th packet generated using Blockwise
code $k$ by $P_{k,i}$. We transmit the packets of the two codes
in an interleaved manner (in the sequence $P_{1,1},P_{2,1},P_{1,2},P_{2,2},P_{1,3},...$).
Receiver 1, which uses only the packets generated using Blockwise
code 1, can decode a block using $K$ out of the $L_{1}$ packets
encoded from the block, and therefore can tolerate a packet loss probability
$1-K/L_{1}$. As Blockwise code 1 transmits a block of $K$ packet
per $L_{1}$ channel uses, taking interleaving into account, Receiver
1 can decode at a rate of $K/(2L_{1})$ packets per channel use. Receiver
2 uses packets generated by both codes. It can decode a block using
$K$ out of the $L_{1}+L_{2}$ packets encoded from the block, and
allows a higher packet loss probability $1-K/(L_{1}+L_{2})$. However,
to use the packets generated by both codes, Receiver 2 has to wait
for the slower Blockwise code 2, which transmits a block of $K$ packet
per $L_{2}$ channel uses. Receiver 2 can decode at a rate $K/(2L_{2})$.

\begin{figure}
\centering{}\includegraphics{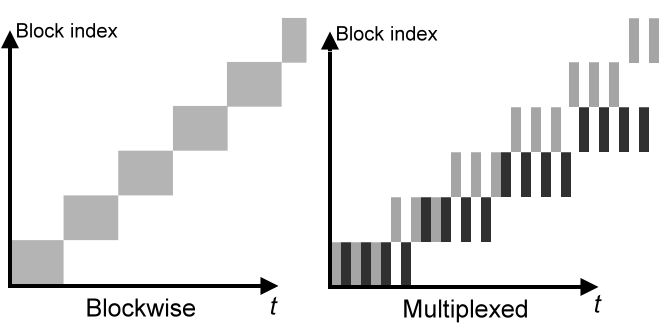}\caption{\label{Figure_block_and_mul_illus}Encoding process of blockwise and
multiplexed codes}
\end{figure}

In the scenario, the receivers with different channel conditions wish
to decode the same sequence of data. Each receiver will decode the
data sequentially at a roughly constant rate which depends on the
channel condition. We call these settings as \emph{multi-rate sequential
data transmission}. It can be viewed as multilevel diversity coding
\cite{yeung1999} with an additional sequential decoding constraint.
In the following sections, we will discuss various cases of the problem.
The case of one transmitter is described in Section \ref{sec:1D}.
The main contribution lies in Section \ref{sec:Tools} concerning
multiple transmitters, in which multi-rate sequential data transmission
has apparently dissimilar behavior compared to its non-sequential
counterpart.

\section{Formulation of Multi-rate Sequential Data Transmission\label{sec:MRSD_basic}}

We consider the transmission of a sequence of data through an erasure
channel. To simplify the setting, we assume that one bit is sent at
a time through the erasure channel. Assume the sender is going to
transmit a sequence of bits $M_{1},M_{2},...$. We assume $M_{i}\stackrel{i.i.d.}{\sim}\mathrm{Bern}(1/2)$.
The sender will encode it into a sequence of bits $X_{1},X_{2},...$
and transmit them through an erasure channel. The symbols arrive at
the receiver as $Y_{1},Y_{2},...$, where some of them may be erased
(denoted by $e$, $Y_{i}\in\{0,1,e\}$). Based on these symbols, the
receiver tries to decode $M_{1},M_{2},...$.
\begin{defn}
[MRS code] A \emph{multi-rate sequential} \emph{code} (MRS code)
is specified by a pair of encoding and decoding functions. The encoding
function is a function mapping the message $\left\{ M_{i}\right\} $
to the encoding symbols $\left\{ X_{j}\right\} $. There is a random
variable $Q$ supported in $\mathcal{Q}$ which is known by both the
sender and the receiver, and independent of the message and the channel
erasure, to allow random coding scheme. The encoding function is a
function 
\begin{eqnarray*}
\mathrm{Enc}\,:\,\mathcal{Q}\times\{0,1\}^{\mathbb{N}}\times\mathbb{N} & \rightarrow & \{0,1\}\\
\left(Q,\left\{ M_{i}\right\} ,j\right) & \mapsto & X_{j}.
\end{eqnarray*}

Note that this definition allows the encoder to look at all blocks.

The decoding function maps the received symbols $\left\{ Y_{j}\right\} $,
where some of them may be erased, to the recovered blocks $\left\{ \widetilde{M}_{i}\right\} $.
Let $\mathcal{Y}=\bigcup_{n\in\mathbb{N}}\{0,1,e\}^{n}$ be the space
of received symbols. The decoding function is a function
\begin{eqnarray*}
\mathrm{Dec}\,:\,\mathcal{Q}\times\mathcal{Y}\times\mathbb{N} & \rightarrow & \{0,1\}\\
\left(Q,Y_{1}^{n},i\right) & \mapsto & \widetilde{M}_{i}.
\end{eqnarray*}

We use the notation $Y_{a}^{b}=(Y_{a},Y_{a+1},...,Y_{b})$. For simplicity,
we write $X_{j}\left(Q,\left\{ M_{i}\right\} \right))=\mathrm{Enc}\left(Q,\left\{ M_{i}\right\} ,j\right)$,
and $\widetilde{M}_{i}\left(Q,Y_{1}^{n}\right)=\mathrm{Dec}\left(Q,Y_{1}^{n},i\right)$.
\end{defn}

The MRS code does not admit a fixed rate like other block codes. Instead
its rate depends on the channel capacity.
\begin{defn}
[admissible pair] A rate-capacity pair $(r,c)$ is called \emph{$\epsilon$-admissible}
by a code if there exist $N_{0}$ such that when $X_{i}\to Y_{i}$
is an erasure channel with capacity $c$ (i.e. erasure channel with
erasure probability $1-c$), 
\[
\mathbb{P}\left\{ M_{m}\neq\widetilde{M}_{m}\left(Q,Y_{1}^{N}\right)\right\} \le\epsilon\:\text{ for any }N\ge N_{0}\text{ and }m\le N(r-\epsilon).
\]

In other words, any receiver with channel capacity $c$ can decode
the first $N(r-\epsilon)$ bits $M_{1}^{N(r-\epsilon)}$ with bit
error probability less than $\epsilon$ when the first $N$ symbols
$Y_{1}^{N}$ are received, for sufficiently large $N$.
\end{defn}

It is clear that if $(c,r)$ is \emph{$\epsilon$}-admissible, then
all pairs in $\left\{ (c',r')|c'\ge c,r'\le r\right\} $ are \emph{$\epsilon$}-admissible.
Therefore we can use a function to characterize the rate of a code.
We call $r:[0,1]\to[0,\infty)$ a \emph{rate-capacity function} if
it is monotonically increasing, right continuous, and there exists
an $\eta>0$ such that $r(c)=0$ for $c\le\eta$.
\begin{defn}
[rate of MRS code] \label{def_capratefunc}A rate-capacity function
$r(c)$ is called \emph{$\epsilon$-admissible} by a code if all of
the pairs $(c,r(c))$ are \emph{$\epsilon$}-admissible by the code.
\end{defn}

The rate-capacity functions of the blockwise code and the multiplexed
blockwise code described in the introduction can be given by Figure
\ref{Figure_blockmul}.

\begin{figure}
\centering{}\includegraphics[scale=1.5]{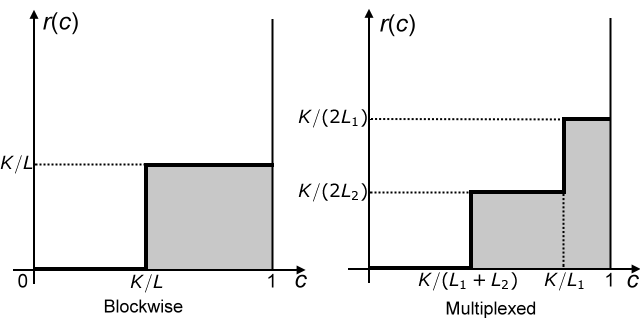}\caption{\label{Figure_blockmul}Rate-capacity functions of blockwise and multiplexed
blockwise code}
\end{figure}

\begin{defn}
[achievable rate-capacity functions] A rate-capacity function $r(c)$
is \emph{achievable} if for any $\epsilon>0$, there exist a code
where $r(c)$ is \emph{$\epsilon-$}admissible by that code.
\end{defn}

\section{Superposition Coding for Single Transmitter\label{sec:1D}}

In this section, we will present the design of superposition multi-rate
sequential codes, and prove their optimality by giving the set of
achievable rate-capacity functions explicitly.
\begin{defn}
[superposition MRS code] A superposition MRS code is characterized
by the block size $K$ and the parameter $g:[0,\infty)\to[0,\infty)$
which is bounded, monotonically decreasing and left continuous with
$\int_{0}^{\infty}g(\alpha)d\alpha=1$. The message $\left\{ M_{i}\right\} $
is divided into blocks of $K$ bits, $B_{i}=M_{(i-1)K+1}^{iK}$. In
the encoding process of the code, we first sample a sequence of non-negative
random variables $A_{1},A_{2},...$ i.i.d. according to the cumulative
distribution function

\begin{eqnarray}
F_{A}(\alpha) & = & \int_{0}^{\alpha}g(x)dx-\alpha g(\alpha)\nonumber \\
 & = & -\int_{0}^{\alpha}xdg(x).\label{eq:FA_def}
\end{eqnarray}
Note that $F_{A}(\alpha)$ is increasing as $g(\alpha)$ is decreasing.
The sequence is known by both the sender and the receiver (we may
let $Q=(A_{1},A_{2},...)$). At time instance $i$, the sender generate
an encoding symbol $X_{i}$ from the block $B_{\left\lceil i\cdot A_{i}/K\right\rceil }$
using random linear projections and transmit it to the receiver.
\end{defn}
Note that we employ a random coding scheme. Random linear projections
allow us to decode a block of $K$ bits using $K+o(K)$ encoding bits
with an arbitrarily small error probability. The superposition MRS
code is essentially performing time multiplex on blockwise codes at
different rates.

The parameter $g(\alpha)$ roughly corresponds to the proportion of
the encoding symbols dedicated to satisfying the need of the receivers
which wish to decode at rate $\alpha$. When we send encoding symbols
encoded at rate $\gamma$ (i.e. $A_{i}=\gamma$), a receiver which
decodes at rate $\alpha\le\gamma$ can use a portion of $\alpha/\gamma$
of the encoding symbols. The parameter $g(\alpha)$ describes the
proportion of the symbols which can be used, divided by the proportion
needed (which is $\alpha$), given by
\[
g(\alpha)=\frac{1}{\alpha}\int_{\alpha}^{\infty}\frac{\alpha}{\gamma}dF_{A}(\gamma),
\]
which can be verified using (\ref{eq:FA_def}):
\begin{eqnarray*}
\frac{1}{\alpha}\int_{\alpha}^{\infty}\frac{\alpha}{\gamma}dF_{A}(\gamma) & = & \int_{\alpha}^{\infty}\frac{1}{\gamma}dF_{A}(\gamma)\\
 & = & -\int_{\alpha}^{\infty}\frac{1}{\gamma}\cdot\gamma dg(\gamma)\\
 & = & -\int_{\alpha}^{\infty}dg(\gamma)=g(\alpha).
\end{eqnarray*}

The blockwise code and the multiplexed blockwise code described in
the introduction are examples of superposition MRS codes. For the
blockwise code, the parameter is taken to be 
\[
g(\alpha)=\begin{cases}
L/K & \text{ when }\alpha\le K/L\\
0 & \text{ when }\alpha>K/L.
\end{cases}
\]

For the multiplexed blockwise code, the parameter is taken to be 
\[
g(\alpha)=\begin{cases}
\frac{L_{1}+L_{2}}{K} & \text{ when }\alpha\le\frac{K}{2L_{2}}\\
\frac{L_{1}}{K} & \text{ when }\frac{K}{2L_{2}}<\alpha\le\frac{K}{2L_{1}}\\
0 & \text{ when }\alpha>\frac{K}{2L_{1}}.
\end{cases}
\]

The choice of the parameter is closely related to the rate-capacity
function we would like to achieve. The following theorem describes
the relationship between the two functions.
\begin{thm}
\label{thm:1Dsupregion}For a fixed $\epsilon$, the rate-capacity
function $r(c)$ is $\epsilon$-admissible by the superposition MRS
code with block size $K$ and parameter $g(\alpha)$ for all sufficiently
large $K$, if there exists $\xi>0$ satisfying
\[
c\cdot g\left(r(c)\right)\ge1+\xi\text{ for all }c>0\text{ with }r(c)>0.
\]
\end{thm}
\begin{proof}
Fix $\epsilon>0$. Let $r(c)$ be a rate-capacity function, and let
$g(\alpha)$ be a bounded and monotonically decreasing function (let
$g(\alpha)\le\zeta$ for all $\alpha\ge0$). Assume the condition
$c\cdot g\left(r(c)\right)\ge1+\xi$ is satisfied for some $\xi>0$.
We take $K\ge\max\left(4/\xi,\,32\cdot(\zeta+1)/(\epsilon\xi)\right)$.
We now consider the superposition MRS code with block size $K$ and
parameter $g(\alpha)$.

Fix any channel capacity $c$ with $r(c)>0$. Let $\alpha_{1}=r(c)$,
and $\alpha_{0}=r(c)-\epsilon$. At time instance $i$, the sender
generate an encoding symbol from the block $B_{\left\lceil i\cdot A_{i}/K\right\rceil }$
(note that $B_{i}=M_{(i-1)K+1}^{iK}$). Let $F_{A}(\alpha)=\int_{0}^{\alpha}g(x)dx-\alpha g(\alpha)$.
The probability that $B_{j}$ is chosen is
\[
\mathbb{P}\left\{ A_{i}\in\left(\frac{K(j-1)}{i},\frac{Kj}{i}\right]\right\} =F_{A}\left(\frac{Kj}{i}\right)-F_{A}\left(\frac{K(j-1)}{i}\right).
\]

Let $n\ge\frac{K}{\epsilon}$, and $j\le\frac{n\alpha_{0}}{K}+1$.
We will study whether the block $B_{j}$ can be decoded using $Y_{1}^{n}$
with error probability less than $\epsilon$ when the channel capacity
is $c$. If so, then $M_{1}^{\left\lfloor \alpha_{0}n\right\rfloor }$
can be decoded using $Y_{1}^{n}$ with bit error probability less
than $\epsilon$ whenever $n\ge\frac{K}{\epsilon}$, and thus the
rate-capacity pair $(\alpha_{1},c)$ is $\epsilon-$admissible. Let
$S$ be the random variable representing the number of times $B_{j}$
is chosen in the encoding of $X_{i}$, and $X_{i}$ is not erased,
for $i=1,...,n$. Its expected value is given by
\begin{eqnarray*}
\mathbb{E}[S] & = & c\cdot\sum_{i=1}^{n}\left(F_{A}\left(\frac{Kj}{i}\right)-F_{A}\left(\frac{K(j-1)}{i}\right)\right).
\end{eqnarray*}

Note that $F_{A}$ is monotonically increasing and not greater than
1 (and therefore $F_{A}\left(\frac{Kj}{nx}\right)$ is monotonically
decreasing with respect to $x$), we have
\begin{eqnarray*}
\sum_{i=1}^{n}F_{A}\left(\frac{Kj}{i}\right) & \ge & n\left(\int_{1/n}^{1}F_{A}\left(\frac{Kj}{nx}\right)dx\right)\\
 & \ge & n\left(\int_{0}^{1}F_{A}\left(\frac{Kj}{nx}\right)dx\right)-1,
\end{eqnarray*}
\begin{eqnarray*}
\sum_{i=1}^{n}F_{A}\left(\frac{K(j-1)}{i}\right) & \le & n\left(\int_{0}^{1}F_{A}\left(\frac{K(j-1)}{nx}\right)dx\right).
\end{eqnarray*}

Therefore
\begin{eqnarray*}
\mathbb{E}[S] & \ge & cn\left(\int_{0}^{1}\left(F_{A}\left(\frac{Kj}{nx}\right)-F_{A}\left(\frac{K(j-1)}{nx}\right)\right)dx\right)-c\\
 & \ge & cn\left(\int_{0}^{1}\left(F_{A}\left(\frac{Kj}{nx}\right)-F_{A}\left(\frac{K(j-1)}{nx}\right)\right)dx\right)-1,
\end{eqnarray*}
where
\begin{eqnarray*}
\int_{0}^{1}F_{A}\left(\frac{\alpha}{x}\right)dx & = & F_{A}(\alpha)-\int_{0}^{1}xdF_{A}\left(\frac{\alpha}{x}\right)\\
 & = & F_{A}(\alpha)+\alpha\int_{\alpha}^{\infty}\frac{1}{x}dF_{A}\left(x\right)\\
 & = & F_{A}(\alpha)+\alpha\int_{\alpha}^{\infty}\frac{1}{x}\cdot-xdg(x)\\
 & = & F_{A}(\alpha)-\alpha\int_{\alpha}^{\infty}1dg(x)\\
 & = & F_{A}(\alpha)+\alpha g(\alpha)\\
 & = & \int_{0}^{\alpha}g(x)dx.
\end{eqnarray*}

Hence, as $g(x)$ is monotonically decreasing,
\begin{eqnarray*}
\mathbb{E}[S] & \ge & cn\cdot\left(\int_{0}^{Kj/n}g(x)dx-\int_{0}^{K(j-1)/n}g(x)dx\right)-1\\
 & = & cn\cdot\int_{K(j-1)/n}^{Kj/n}g(x)dx-1\\
 & \ge & cK\cdot g\left(\frac{Kj}{n}\right)-1\\
 & \ge & cK\cdot g\left(\alpha_{0}+\frac{K}{n}\right)-1\\
 & \ge & cK\cdot g\left(\alpha_{0}+\epsilon\right)-1\\
 & = & cK\cdot g\left(\alpha_{1}\right)-1\\
 & \ge & K(1+\xi)-1.
\end{eqnarray*}

On the other hand,
\begin{eqnarray*}
\sum_{i=1}^{n}F_{A}\left(\frac{Kj}{i}\right) & \le & n\left(\int_{0}^{1}F_{A}\left(\frac{Kj}{nx}\right)dx\right),
\end{eqnarray*}
and
\begin{eqnarray*}
\sum_{i=1}^{n}F_{A}\left(\frac{K(j-1)}{i}\right) & \ge & n\left(\int_{1/n}^{1}F_{A}\left(\frac{K(j-1)}{nx}\right)dx\right)\\
 & \ge & n\left(\int_{0}^{1}F_{A}\left(\frac{K(j-1)}{nx}\right)dx\right)-1.
\end{eqnarray*}

Let 
\[
p_{i}=c\left(F_{A}\left(\frac{Kj}{i}\right)-F_{A}\left(\frac{K(j-1)}{i}\right)\right).
\]

The variance of $S$ can be given by
\begin{eqnarray*}
\mathrm{Var}[S] & = & \sum_{i=1}^{n}\left(p_{i}-p_{i}^{2}\right)\\
 & \le & \sum_{i=1}^{n}p_{i}\\
 & \le & cn\cdot\left(\int_{0}^{1}\left(F_{A}\left(\frac{Kj}{nx}\right)-F_{A}\left(\frac{K(j-1)}{nx}\right)\right)dx\right)+c\\
 & \le & n\cdot\left(\int_{0}^{1}\left(F_{A}\left(\frac{Kj}{nx}\right)-F_{A}\left(\frac{K(j-1)}{nx}\right)\right)dx\right)+1\\
 & = & n\cdot\left(\int_{0}^{Kj/n}g_{2}(x)dx-\int_{0}^{K(j-1)/n}g_{2}(x)dx\right)+1\\
 & = & n\cdot\int_{K(j-1)/n}^{Kj/n}g_{2}(x)dx+1\\
 & \le & K\zeta+1,
\end{eqnarray*}
\[
\sigma_{S}\le\sqrt{K}\cdot\sqrt{\zeta+1}.
\]

As a result, for $K$ large enough, $S$ is close to $\mathbb{E}[S]$
with high probability. Chebyshev's inequality gives
\[
\mathbb{P}\left\{ S<K(1+\xi)-1-\frac{\sqrt{K}\cdot\sqrt{\zeta+1}}{\sqrt{\epsilon/2}}\right\} \le\epsilon/2,
\]
where
\begin{eqnarray*}
\lefteqn{K(1+\xi)-1-\frac{\sqrt{K}\cdot\sqrt{\zeta+1}}{\sqrt{\epsilon/2}}}\\
 & = & K\left(1+\xi-\frac{1}{K}-\frac{\sqrt{\zeta+1}}{\sqrt{K\epsilon/2}}\right)\\
 & \ge & K\left(1+\xi/2\right)
\end{eqnarray*}
by the assumption $K\ge\max\left(4/\xi,\,32\cdot(\zeta+1)/(\epsilon\xi)\right)$.
When random linear projections are used, for a fixed probability of
error $\epsilon/2$, the number of symbols needed to decode a block
of $K$ bits is $K+o(K)$, which is smaller than $K\left(1+\xi/2\right)$
for sufficiently large $K$. Therefore the rate function $r(c)$ is
$\epsilon$-admissable by the superposition MRS code with parameter
$g(\alpha)$ for $K$ large enough.
\end{proof}
The following theorem gives the achievable region for the rate-capacity
function, which coincides with the region attained by superposition
MRS code.
\begin{thm}
\label{thm:1Diff}A rate-capacity function $r(c)$ is achievable if
and only if
\[
\int_{0}^{1}\frac{1}{c}dr(c)\le1.
\]
\end{thm}
\begin{proof}
[Proof of achievability] Assume $r(c)$ is a rate-capacity function
satisfying $\int_{0}^{1}\frac{1}{c}dr(c)\le1$. We can find $\eta>0$
such that $r(c)=0$ for $c\le\eta$. Let $g(\alpha)=1/r^{-1}(\alpha)$,
where $r^{-1}(\alpha)=\inf\{c|r(c)\ge\alpha\}$. Then $g(\alpha)$
is monotonically decreasing and bounded above by $1/\eta$. It is
implied by $\int_{0}^{1}\frac{1}{c}dr(c)\le1$ that $\int_{0}^{\infty}g(\alpha)d\alpha\le1$.
As $r^{-1}(r(c))\le c$, we also have $c\cdot g\left(r(c)\right)\ge1$
for all $c>0$ with $r(c)>0$.

Fix $\epsilon>0$. Note that $\int_{0}^{\infty}g(\alpha+\epsilon/2)d\alpha<1$.
We can define a function $g_{2}(\alpha)$ by
\[
g_{2}(\alpha)=(1+\xi)\cdot g(\alpha+\epsilon/2),
\]
where $\xi>0$ such that $\int_{0}^{\infty}g_{2}(\alpha)d\alpha=1$.
We know $g_{2}$ is bounded above by $(1+\xi)/\eta$. Define a new
rate-capacity function $r_{2}(c)=\max(r(c)-\epsilon/2,0)$. We have,
for any $c$ with $r_{2}(c)>0$,
\[
c\cdot g_{2}\left(r_{2}(c)\right)=c\cdot(1+\xi)g\left(r(c)\right)\ge1+\xi.
\]

Consider the superposition MRS code with parameter $g_{2}(\alpha)$.
By Theorem \ref{thm:1Dsupregion}, the rate-capacity function $r_{2}(c)$
is $\epsilon/2$-admissible by the code for sufficiently large block
size, which implies that $r(c)$ is $\epsilon$-admissible by the
code.
\end{proof}
\[
\,
\]

\begin{proof}
[Proof of converse] The proof employs a similar idea as in \cite{albanese1996}.
As $r(c)$ is a rate-capacity function, we can find $\eta>0$ such
that $r(c)=0$ for $c\le\eta$. For any $\epsilon>0$, consider a
code where $r(c)$ is \emph{$\epsilon$}-admissible. The message $\{M_{i}\}$
are encoded into binary symbols $\{X_{i}\}$, and sent through an
erasure channel with capacity $c$ (we call it Channel $c$) to give
$\{Y_{c,i}\}$ for all $c>0$. Assume we have the following for any
$c$,
\[
\max_{m\le n(r(c)-\epsilon)}\mathbb{P}\left\{ M_{m}\neq\widetilde{M}_{m}\left(Q,Y_{c,1}^{n}\right)\right\} <\epsilon\:\text{ for any }n\ge N_{0}.
\]

Let $N\ge N_{0}$. Let $E_{c,i}$ be the indicator of the events of
erasure in Channel $c$ ($E_{c,i}=1$ means that $Y_{c,i}=e$). Note
that whether the message can be decoded at the receiver with Channel
$c$ depends only on the marginal distribution of $\left\{ E_{c,i},X_{i},M_{i},Q\right\} _{i\in\mathbb{N}}$,
and is conditional independent of $E_{c_{2},i}$ for $c_{2}\neq c$
given $\left\{ E_{c,i},X_{i},M_{i},Q\right\} _{i\in\mathbb{N}}$.
Therefore, we may modify the joint distribution of $\left\{ E_{c,i}\right\} _{c\in(0,1],i\in\mathbb{N}}$
without affecting the result as long as the distributions of $\left\{ E_{c,i},X_{i},M_{i},Q\right\} _{i\in\mathbb{N}}$
are preserved. From now on, we assume the channels are cascaded, i.e.
whenever $c_{0}<c_{1}\le1$, the Markov chain $X_{i}\rightarrow Y_{c_{1},i}\rightarrow Y_{c_{0},i}$
holds, and 
\[
E_{c_{0},i}=\begin{cases}
E_{c_{1},i} & \text{ w.p. }c_{0}/c_{1}\\
1 & \text{ w.p. }1-c_{0}/c_{1}.
\end{cases}
\]

Let 
\begin{equation}
r_{2}(c)=\max(r(c)-\epsilon,0).\label{eq:1Diff_Def_r2}
\end{equation}

Define
\begin{equation}
f(c)=\frac{1}{Nc}\cdot H(Y_{c,1}^{N}|Q,M_{1}^{\left\lfloor Nr_{2}(c)\right\rfloor },E_{c,1}^{N}).\label{eq:1Diff_Def_f(c)}
\end{equation}

The common random variable $Q$ will be omitted for simplicity. Consider
Channel $c_{0}$ and Channel $c_{1}$ where $\eta/2\le c_{0}<c_{1}\le1$.
Let $k_{0}=\left\lfloor Nr_{2}(c_{0})\right\rfloor $, and $k_{1}=\left\lfloor Nr_{2}(c_{1})\right\rfloor $.
Note that
\begin{eqnarray}
\lefteqn{H(Y_{c_{1},1}^{N}|M_{1}^{k_{1}},E_{c_{1},1}^{N})}\nonumber \\
 & = & H(Y_{c_{1},1}^{N},M_{k_{0}+1}^{k_{1}}|M_{1}^{k_{0}},E_{c_{1},1}^{N})-H(M_{k_{0}+1}^{k_{1}}|M_{1}^{k_{0}},E_{c_{1},1}^{N}),\label{eq:1Diff_sum_term}
\end{eqnarray}
where, due to the assumption that $M_{i}$ are i.i.d. uniform in $\{0,1\}$,
\[
H(M_{k_{0}+1}^{k_{1}}|M_{1}^{k_{0}},E_{c_{1},1}^{N})=H(M_{k_{0}+1}^{k_{1}})=k_{1}-k_{0}\ge N(r_{2}(c_{1})-r_{2}(c_{0}))-1.
\]

As $M_{1}^{k_{1}}$ can be decoded using $Y_{c_{1},1}^{N}$ with bit
error probability less than $\epsilon$, by Fano's inequality,
\begin{eqnarray*}
\lefteqn{H(Y_{c_{1},1}^{N},M_{k_{0}+1}^{k_{1}}|M_{1}^{k_{0}},E_{c_{1},1}^{N})}\\
 & = & H(Y_{c_{1},1}^{N}|M_{1}^{k_{0}},E_{c_{1},1}^{N})\\
 &  & \;+H(M_{k_{0}+1}^{k_{1}}|Y_{c_{1},1}^{N},M_{1}^{k_{0}},E_{c_{1},1}^{N})\\
 & \le & H(Y_{c_{1},1}^{N}|M_{1}^{k_{0}},E_{c_{1},1}^{N})\\
 &  & \;+\sum_{i=k_{0}+1}^{k_{1}}H(M_{i}|Y_{c_{1},1}^{N},M_{1}^{k_{0}},E_{c_{1},1}^{N})\\
 & \le & H(Y_{c_{1},1}^{N}|M_{1}^{k_{0}},E_{c_{1},1}^{N})+(k_{1}-k_{0})\cdot H(\epsilon)\\
 & \le & H(Y_{c_{1},1}^{N}|M_{1}^{k_{0}},E_{c_{1},1}^{N})+\left(N(r_{2}(c_{1})-r_{2}(c_{0}))+1\right)\cdot H(\epsilon)\\
 & \le & H(Y_{c_{1},1}^{N}|M_{1}^{k_{0}},E_{c_{1},1}^{N})+\left(N(r_{2}(c_{1})-r_{2}(c_{0}))\right)\cdot H(\epsilon)+2.
\end{eqnarray*}

Let $E'_{i}\stackrel{i.i.d.}{\sim}\mathrm{Bern}(1-c_{0}/c_{1})$ with
\[
E_{c_{0},i}=\begin{cases}
E_{c_{1},i} & \text{ if }E'_{i}=0\\
1 & \text{ if }E'_{i}=1.
\end{cases}
\]

We can obtain

\begin{eqnarray*}
\lefteqn{H(Y_{c_{0},1}^{N}|M_{1}^{k_{0}},E_{c_{0},1}^{N})}\\
 & = & \sum_{i=1}^{N}H(Y_{c_{0},i}|M_{1}^{k_{0}},Y_{c_{0},1}^{i-1},E_{c_{0},1}^{N})\\
 & \stackrel{\mathrm{(i)}}{\ge} & \sum_{i=1}^{N}H(Y_{c_{0},i}|M_{1}^{k_{0}},Y_{c_{1},1}^{i-1},E_{c_{0},1}^{N})\\
 & \stackrel{\mathrm{(ii)}}{=} & \sum_{i=1}^{N}H(Y_{c_{0},i}|M_{1}^{k_{0}},Y_{c_{1},1}^{i-1},E_{c_{0},i})\\
 & \ge & \sum_{i=1}^{N}H(Y_{c_{0},i}|M_{1}^{k_{0}},Y_{c_{1},1}^{i-1},E_{c_{0},i},E'_{i})\\
 & \stackrel{\mathrm{(iii)}}{=} & \sum_{i=1}^{N}\left(\frac{c_{0}}{c_{1}}\cdot H(Y_{c_{0},i}|M_{1}^{k_{0}},Y_{c_{1},1}^{i-1},E_{c_{0},i},E'_{i}=0)\right)\\
 & = & \sum_{i=1}^{N}\left(\frac{c_{0}}{c_{1}}\cdot H(Y_{c_{1},i}|M_{1}^{k_{0}},Y_{c_{1},1}^{i-1},E_{c_{1},i},E'_{i}=0)\right)\\
 & \stackrel{\mathrm{(iv)}}{=} & \sum_{i=1}^{N}\left(\frac{c_{0}}{c_{1}}\cdot H(Y_{c_{1},i}|M_{1}^{k_{0}},Y_{c_{1},1}^{i-1},E_{c_{1},i})\right)\\
 & \stackrel{\mathrm{(v)}}{\ge} & \sum_{i=1}^{N}\left(\frac{c_{0}}{c_{1}}\cdot H(Y_{c_{1},i}|M_{1}^{k_{0}},Y_{c_{1},1}^{i-1},E_{c_{1},1}^{N})\right)\\
 & = & \frac{c_{0}}{c_{1}}\cdot H(Y_{c_{1},1}^{N}|M_{1}^{k_{0}},E_{c_{1},1}^{N}),
\end{eqnarray*}
where (i) is due to $H(Y_{c_{0},1}^{i-1}|Y_{c_{1},1}^{i-1},E_{c_{0},1}^{N})=0$,
(ii) is due to $\left.(E_{c_{0},1}^{i-1},E_{c_{0},i+1}^{N})\left.\perp\!\!\!\perp\right.Y_{c_{0},i}\right|(M_{1}^{k_{0}},Y_{c_{1},1}^{i-1})$,
(iii) is obtained by conditioning on $E'_{i}$ and by $Y_{c_{0},i}=e$
when $E'_{i}=1$, (iv) is due to $E'_{i}\left.\perp\!\!\!\perp\right.(Y_{c_{1},i},M_{1}^{k_{0}},Y_{c_{0},1}^{i-1},E_{c_{1},i})$,
and (v) is due to $H(E_{c_{1},1}^{i-1}|Y_{c_{1},1}^{i-1})=0$ and
$\left.E_{c_{1},i+1}^{N}\left.\perp\!\!\!\perp\right.Y_{c_{1},i}\right|(M_{1}^{k_{0}},Y_{c_{1},1}^{i-1})$.

Hence by (\ref{eq:1Diff_sum_term}),

\begin{eqnarray*}
\lefteqn{H(Y_{c_{1},1}^{N}|M_{1}^{k_{1}},E_{c_{1},1}^{N})}\\
 & \le & \frac{c_{1}}{c_{0}}\cdot H(Y_{c_{0},1}^{N}|M_{1}^{k_{0}},E_{c_{0},1}^{N})+\left(N(r_{2}(c_{1})-r_{2}(c_{0}))\right)\cdot H(\epsilon)+2\\
 &  & \;-\left(N(r_{2}(c_{1})-r_{2}(c_{0}))-1\right),
\end{eqnarray*}

After replacing the terms by $f(c)$ using (\ref{eq:1Diff_Def_f(c)}),
\[
\frac{1}{c_{1}}\left(r_{2}(c_{1})-r_{2}(c_{0})\right)\le\frac{f(c_{0})-f(c_{1})+\frac{3}{Nc_{1}}}{1-H(\epsilon)}\le\frac{f(c_{0})-f(c_{1})+\frac{6}{N\eta}}{1-H(\epsilon)}.
\]

Due to the monotonicity of $r_{2}(c)$,
\begin{eqnarray*}
\lefteqn{\int_{c_{0}}^{c_{1}}\frac{1}{c}dr_{2}(c)-\left(\frac{1}{c_{0}}-\frac{1}{c_{1}}\right)\left(r_{2}(c_{1})-r_{2}(c_{0})\right)}\\
 & \le & \frac{1}{c_{0}}\left(r_{2}(c_{1})-r_{2}(c_{0})\right)-\left(\frac{1}{c_{0}}-\frac{1}{c_{1}}\right)\left(r_{2}(c_{1})-r_{2}(c_{0})\right)\\
 & = & \frac{1}{c_{1}}\left(r_{2}(c_{1})-r_{2}(c_{0})\right)\\
 & \le & \frac{f(c_{0})-f(c_{1})+\frac{6}{N\eta}}{1-H(\epsilon)}.
\end{eqnarray*}

Let $m=\left\lfloor \sqrt{N}\right\rfloor $, and $c_{0}=c'_{0}<c'_{1}<...<c'_{m}=c_{1}$
such that $r_{2}(c)$ is continuous at $c=c'_{1},...,c'_{m-1}$ and
\[
\frac{1}{c'_{i}}-\frac{1}{c'_{i+1}}<\frac{2}{m}\left(\frac{1}{c_{0}}-\frac{1}{c_{1}}\right)\text{ for }i=0,...,m-1.
\]

We can always find such $c'_{i}$ as a monotonic function has at most
countable discontinuities. Note that for $i=0,...,m-1$,
\[
\frac{1}{c'_{i}}-\frac{1}{c'_{i+1}}<\frac{2}{m}\left(\frac{1}{c_{0}}-\frac{1}{c_{1}}\right)\le\frac{4}{m\eta}
\]
as $\eta/2\le c_{0}<c_{1}$. Then we have
\begin{eqnarray*}
\lefteqn{\int_{c_{0}}^{c_{1}}\frac{1}{c}dr_{2}(c)}\\
 & = & \sum_{i=0}^{m-1}\int_{c'_{i}}^{c'_{i+1}}\frac{1}{c}dr_{2}(c)\\
 & \le & \sum_{i=0}^{m-1}\left(\frac{f(c'_{i})-f(c'_{i+1})+\frac{6}{N\eta}}{1-H(\epsilon)}+\left(\frac{1}{c'_{i}}-\frac{1}{c'_{i+1}}\right)\left(r_{2}(c'_{i+1})-r_{2}(c'_{i})\right)\right)\\
 & \le & \frac{f(c_{0})-f(c_{1})+\frac{6m}{N\eta}}{1-H(\epsilon)}+\frac{4}{m\eta}\left(r_{2}(c_{1})-r_{2}(c_{0})\right)\\
 & \le & \frac{f(c_{0})-f(c_{1})+\frac{6m}{N\eta}}{1-H(\epsilon)}+\frac{4}{m\eta}r(1).
\end{eqnarray*}

Note that
\begin{eqnarray*}
f(c) & = & \frac{1}{Nc}\cdot H(Y_{c,1}^{N}|Q,M_{1}^{\left\lfloor N(r(c)-\epsilon)\right\rfloor },E_{c,1}^{N})\\
 & \le & \frac{1}{Nc}\cdot H(Y_{c,1}^{N}|,E_{c,1}^{N})\\
 & \le & 1.
\end{eqnarray*}

Thus we have, for any $\eta/2\le c_{0}<c_{1}\le1$,
\[
\int_{c_{0}}^{c_{1}}\frac{1}{c}dr_{2}(c)\le\frac{1+\frac{6m}{N\eta}}{1-H(\epsilon)}+\frac{4}{m\eta}r(1).
\]

Therefore we can obtain an inequality on $r(c)$ using (\ref{eq:1Diff_Def_r2})
by 
\begin{eqnarray*}
\int_{c_{0}}^{c_{1}}\frac{1}{c}dr(c) & = & \int_{c_{0}}^{c_{1}}\frac{1}{c}d\left(r(c)-\epsilon\right)\\
 & \le & \int_{c_{0}}^{c_{1}}\frac{1}{c}dr_{2}(c)+\frac{\epsilon}{\eta}\\
 & \le & \frac{1+\frac{6m}{N\eta}}{1-H(\epsilon)}+\frac{4}{m\eta}r(1)+\frac{\epsilon}{\eta}.
\end{eqnarray*}

The inequality holds for arbitrarily large $N$ and arbitrarily small
$\epsilon$. We can conclude that
\[
\int_{0}^{1}\frac{1}{c}dr(c)=\int_{\eta/2}^{1}\frac{1}{c}dr(c)\le1.
\]
\end{proof}
\begin{rem*}
It is shown in the theorem that the achievable region of MRS codes
coincides with that of priority encoding transmission or multilevel
diversity coding. When there are more than one transmitters, the MRS
codes no longer admit the same region as priority encoding transmission
or multilevel diversity coding in general.
\end{rem*}

\section{The Multiple Transmitter Setting\label{sec:MultiD}}

In this section, we will discuss the case where there are $d$ transmitters
which cooperate to send the same sequence of data $\left\{ M_{i}\right\} $,
but they may or may not be transmitting the same sequence of encoding
symbols. There are multiple receivers that wish to decode $\left\{ M_{i}\right\} $
sequentially at different rates. A receiver has an erasure channel
connected to each of the $d$ transmitters, and the channels may have
different capacities. The central question of this section is that,
given a set of receivers with different capacities and different rate
requirements, is it possible to design a code which can satisfy the
need of all receivers?

We call the transmitters as Transmitter $k$, where $k=1,...,d$.
We denote the symbol sent by Transmitter $k$ at time $n$ by $X_{k,n}\in\{0,1\}$.
For a receiver with a channel from Transmitter $k$ with capacity
$c$, denote the symbol received from Transmitter $k$ by $Y_{k,c,n}\in\{0,1,e\}$,
and the indicator of erasure $E_{k,c,n}\in\{0,1\}$ ($E_{k,c,n}=1$
indicates an erasure). The definition of a multi-transmitter multi-rate
sequential code is similar to that in the single transmitter case,
and will be omitted.

For a vector of channel capacities $\mathbf{c}\in[0,1]^{d}$, where
$c_{k}$ is the capacity of the channel to Transmitter $k$, we write
\begin{eqnarray*}
X_{n} & = & \left(X_{1,n},...,X_{d,n}\right),\\
Y_{\mathbf{c},n} & = & \left(Y_{1,c_{1},n},...,Y_{d,c_{d},n}\right),\\
E_{\mathbf{c},n} & = & \left(E_{1,c_{1},n},...,E_{d,c_{d},n}\right).
\end{eqnarray*}

We use the notation $Y_{\mathbf{c},a}^{b}=(Y_{\mathbf{c},a},Y_{\mathbf{c},a+1},...,Y_{\mathbf{c},b})$.
From now on, we refer to the receiver with channel capacities $\mathbf{c}$
as Receiver $\mathbf{c}$. We write the sum of capacities in $\mathbf{c}$
by $\Sigma(\mathbf{c})=\sum_{i=1}^{d}c_{i}$.
\begin{defn}
[admissible pair] A rate-capacity pair $(r,\mathbf{c})$ is called
\emph{$\epsilon$-admissible} by a code if Receiver $\mathbf{c}$
can decode the first $N(r-\epsilon)$ bits $M_{1}^{N(r-\epsilon)}$
with bit error probability less than $\epsilon$ when the first $N$
symbols $Y_{\mathbf{c},1}^{N}$ are received, for sufficiently large
$N$. More precisely, there exist $N_{0}$ such that 
\[
\mathbb{P}\left\{ M_{m}\neq\widetilde{M}_{m}\left(Q,Y_{\mathbf{c},1}^{N}\right)\right\} \le\epsilon\:\text{ for any }N\ge N_{0},m\le N(r-\epsilon).
\]

\end{defn}

We use a function to characterize the rate of a code. We call $r:[0,1]^{d}\to[0,\infty)$
a \emph{rate-capacity function} if it is monotonically increasing
and right continuous along each of the $d$ dimensions, and there
exist an $\eta>0$ such that $r(\mathbf{c})=0$ for $\Sigma(\mathbf{c})\le\eta$.
\begin{defn}
[rate of MRS code] A rate-capacity function $r(\mathbf{c})$ is
called \emph{$\epsilon$-admissible} by a code if all of the pairs
$(\mathbf{c},r(\mathbf{c}))$ are \emph{$\epsilon$}-admissible by
the code.
\end{defn}

\begin{defn}
[achievable rate-capacity functions] A rate-capacity function $r(\mathbf{c})$
is \emph{achievable} if for any $\epsilon>0$, there exist a code
where $r(\mathbf{c})$ is \emph{$\epsilon-$}admissible by that code.
\end{defn}

The superposition MRS code for multiple transmitters is similar to
that for single transmitter.
\begin{defn}
[superposition MRS code] A superposition MRS code is characterized
by the block size $K$ and the parameter $g:[0,\infty)\to[0,\infty)^{d}$
which is bounded, monotonically decreasing and left continuous along
each dimension with $\int_{0}^{\infty}g_{k}(\alpha)d\alpha=1$ for
$k=1,...,d$ (write $g_{k}(\alpha)$ for the $k$-th entry of $g(\alpha)$).
Transmitter $k$ generates encoding symbols using the single transmitter
superposition MRS code with block size $K$ and parameter $g_{k}(\alpha)$.
\end{defn}
We say that a rate-capacity function is achievable by superposition
MRS code if it is \emph{$\epsilon-$}admissible by a superposition
MRS code for arbitrarily small $\epsilon$. We give the necessary
and sufficient condition on the achievability by superposition MRS
code.
\begin{thm}
\label{thm:MDsuperIff}The rate-capacity function $r(\mathbf{c})$
is achievable by superposition MRS code if and only if there exists
a function $g:[0,\infty)\to[0,\infty)^{d}$ which is bounded, monotonically
decreasing and left continuous along each of the $d$ dimensions satisfying
\[
\int_{0}^{\infty}g_{k}(\alpha)d\alpha=1\text{ for }k=1,...,d,\text{ and}
\]
\[
\mathbf{c}\cdot g\left(r(\mathbf{c})\right)\ge1\text{ for any }\mathbf{c}\in[0,1]^{d}\text{ with }r(\mathbf{c})>0.
\]
\end{thm}
\begin{proof}
The ``if'' part is similar to the proof of achievability in Theorem
\ref{thm:1Diff}. Fix any $\epsilon>0$. Let $\eta>0$ such that $r(\mathbf{c})=0$
for $\Sigma(\mathbf{c})\le\eta$. Define $g':[0,\infty)\to[0,\infty)^{d}$
by
\[
g'_{k}(\alpha)=(1+\xi_{k})\cdot g_{k}(\alpha+\epsilon/2),
\]
where $\xi_{k}>0$ such that $\int_{0}^{\infty}g'_{k}(\alpha)d\alpha=1$.
Define a new rate-capacity function $r_{2}(\mathbf{c})=\max(r(\mathbf{c})-\epsilon/2,0)$.
We have, for any $\mathbf{c}$ with $r_{2}(\mathbf{c})>0$, 
\[
\mathbf{c}\cdot g'\left(r_{2}(c)\right)=\mathbf{c}\cdot(1+\xi)g'\left(r_{2}(\mathbf{c})\right)\ge1+\xi.
\]

Consider the multiple transmitter superposition MRS code with parameter
$g(\alpha)$. Applying Theorem \ref{thm:1Dsupregion} on each dimension,
the rate-capacity function $r_{2}(\mathbf{c})$ is $\epsilon/2$-admissible
by the code for sufficiently large block size, which implies that
$r(\mathbf{c})$ is $\epsilon$-admissible by the code.

For the ``only if'' part, let $\epsilon>0$ and consider a superposition
MRS code with block size $K$ and parameter $g(\alpha)$ in which
$r(\mathbf{c})$ is $\epsilon$-admissible. Consider the receiver
with capacities $\mathbf{c}$ which decode at rate $r(\mathbf{c})-\epsilon$.
Using similar arguments as in \ref{thm:1Dsupregion}, fixing any block
with sufficiently large index, the expected number of times when a
received symbol from Transmitter $k$ is encoded from the block can
be given by $c_{k}K\cdot g_{k}\left(r(\mathbf{c})-\epsilon\right)$.
As at least $K$ received symbols is required to decode the block,
we have
\[
\mathbf{c}\cdot g\left(r(\mathbf{c})-\epsilon\right)=\sum_{k}c_{k}g_{k}\left(r(\mathbf{c})-\epsilon\right)\ge1.
\]

As $g(\alpha)$ is left continuous, the proof can be completed by
taking $\epsilon\to0$.
\end{proof}

\section{General Non-optimality of Superposition Codes\label{sec:NonOpt}}

We have shown in Section \ref{sec:1D} that superposition codes are
optimal for single transmitter. However, in the multiple transmitter
setting, the superposition MRS codes are not optimal in general. We
will provide a counter example.
\begin{example}
\label{ex_NonOpt}Consider the two transmitter case. Given the block
size $K$ and $0<r<2/3$, the code is constructed by the following
process. For Transmitter $k$ ($k=1,2$), at time $i$, we generate
$X_{k,i}$ by taking a random linear combination of the bits in the
blocks $M'_{2\cdot\left\lceil i/(4K)\right\rceil -2+k}$, $M'_{4\cdot\left\lceil i/(4K)\right\rceil -1-k}$,
and $M'_{4\cdot\left\lceil i/(4K)\right\rceil +1-k}$, where $M'_{i}=M_{(i-1)K+1}^{iK}$.
\end{example}
Consider Transmitter 1, which generates $4K$ encoding bits from the
three blocks $M'_{2\cdot\left\lceil i/(4K)\right\rceil -1}$, $M'_{4\cdot\left\lceil i/(4K)\right\rceil -2}$,
and $M'_{4\cdot\left\lceil i/(4K)\right\rceil }$, with a total of
$3K$ bits. Hence, a receiver which can only receive from Transmitter
1 with channel capacity $3/4$ can decode the three blocks. When $\left\lceil i/(4K)\right\rceil =n$,
the blocks $M'_{2n-1}$, $M'_{4n-2}$, and $M'_{4n}$ can be decoded,
which covers all the blocks. The rate-capacity pair $\left(\frac{1}{2},(\frac{3}{4},0)\right)$
is admissible. Similar for $\left(\frac{1}{2},(0,\frac{3}{4})\right)$.

Consider a receiver which receives from Transmitter 1 and 2, each
with channel capacity $1/2$. It receives $2K$ bits encoding the
three blocks $M'_{2\cdot\left\lceil i/(4K)\right\rceil -1}$, $M'_{4\cdot\left\lceil i/(4K)\right\rceil -2}$,
$M'_{4\cdot\left\lceil i/(4K)\right\rceil }$, and also $2K$ bits
encoding the three blocks $M'_{2\cdot\left\lceil i/(4K)\right\rceil }$,
$M'_{4\cdot\left\lceil i/(4K)\right\rceil -3}$, $M'_{4\cdot\left\lceil i/(4K)\right\rceil -1}$.
When $n\stackrel{def}{=}\left\lceil i/(4K)\right\rceil =1$, the $4K$
bits encoding the blocks $M'_{1}$ to $M'_{4}$ are sufficient to
decode the blocks. When $n\ge2$, assume the blocks $M'_{m}$ for
$m\le4n-4$ are already decoded, then there are $2K$ bits encoding
the two blocks $M'_{4n-2}$ and $M'_{4n}$ ($M'_{2n-1}$ is already
decoded) which are sufficient to decode the blocks, and $2K$ bits
encoding the two blocks $M'_{4n-3}$ and $M'_{4n-1}$ which are sufficient
to decode the blocks.The rate-capacity pair $\left(1,(\frac{1}{2},\frac{1}{2})\right)$
is admissible.

However, the rate-capacity pairs $\left(\frac{1}{2},(\frac{3}{4},0)\right)$
and $\left(1,(\frac{1}{2},\frac{1}{2})\right)$ cannot be simultaneously
achieved by superposition MRS code. If it can be achieved by superposition
MRS code, by Theorem \ref{thm:MDsuperIff}, there is a monotonically
decreasing function $g:[0,\infty)\mapsto[0,\infty)^{2}$ satisfying
$\int_{0}^{\infty}g_{1}(\alpha)d\alpha=\int_{0}^{\infty}g_{2}(\alpha)d\alpha=1$
and $\frac{3}{4}g_{1}(\frac{1}{2})\ge1$, $\frac{1}{2}g_{1}(1)+\frac{1}{2}g_{2}(1)\ge1$.
As $g_{1}(1),g_{2}(1)\le1$, we have $g_{1}(1)=g_{2}(1)=1$, and $g_{1}(\alpha)=g_{2}(\alpha)=1$
when $\alpha\le1$, which contradicts with $\frac{3}{4}g_{1}(\frac{1}{2})\ge1$.

In the following sections, we will study some special cases in which
superposition codes are optimal.

\section{Some Useful Tools\label{sec:Tools}}

We will present some tools which are used to find the admissible region
in certain special cases.

For the sake of simplicity, we write 
\begin{eqnarray*}
J_{N\alpha}^{\infty}(Y_{\mathbf{c},1}^{N}) & = & H(Y_{\mathbf{c},1}^{N}|M_{1}^{\left\lfloor N\alpha\right\rfloor },E_{\mathbf{c},1}^{N},Q),\text{ and}\\
J_{N\alpha}^{N\beta}(Y_{\mathbf{c},1}^{N}) & = & I(M_{\left\lfloor N\alpha\right\rfloor +1}^{\left\lfloor N\beta\right\rfloor };Y_{\mathbf{c},1}^{N}|M_{1}^{\left\lfloor N\alpha\right\rfloor },E_{\mathbf{c},1}^{N},Q).
\end{eqnarray*}

We use the infinity sign ``$\infty$'' in $J_{N\alpha}^{\infty}(Y_{\mathbf{c},1}^{N})$
as the encoding symbols $X_{i}$ are encoded from $M_{1}^{\infty}$
and $Q$, and $Y_{\mathbf{c},i}$ can be determined by $X_{i}$ and
$E_{\mathbf{c},i}$. As a result, 
\[
H(Y_{\mathbf{c},1}^{N}|M_{1}^{\infty},E_{\mathbf{c},1}^{N},Q)=0,
\]
and therefore
\begin{eqnarray*}
J_{N\alpha}^{\infty}(Y_{\mathbf{c},1}^{N}) & = & H(Y_{\mathbf{c},1}^{N}|M_{1}^{\left\lfloor N\alpha\right\rfloor },E_{\mathbf{c},1}^{N},Q)\\
 & = & I(M_{\left\lfloor N\alpha\right\rfloor +1}^{\infty};Y_{\mathbf{c},1}^{N}|M_{1}^{\left\lfloor N\alpha\right\rfloor },E_{\mathbf{c},1}^{N},Q).
\end{eqnarray*}
(Note that the support of $M_{1}^{\infty}$ is uncountable. The above
equations only serve as the intuition behind the definition.)

The quantity $J_{N\alpha}^{N\beta}(Y_{\mathbf{c},1}^{N})$ roughly
corresponds to the amount of information in the first $N$ encoding
symbols dedicated to encode the interval of data $M_{\left\lfloor N\alpha\right\rfloor +1}^{\left\lfloor N\beta\right\rfloor }$.
Furthermore, we define
\begin{eqnarray*}
\overline{J}_{\alpha}^{\infty}(Y_{\mathbf{c}},T) & = & \frac{1}{T}\cdot\int_{0}^{T}J_{\alpha\left\lfloor e^{x}\right\rfloor }^{\infty}(Y_{\mathbf{c},1}^{\left\lfloor e^{x}\right\rfloor })\cdot e^{-x}dx,\\
\overline{J}_{\alpha}^{\beta}(Y_{\mathbf{c}},T) & = & \frac{1}{T}\cdot\int_{0}^{T}J_{\alpha\left\lfloor e^{x}\right\rfloor }^{\beta\left\lfloor e^{x}\right\rfloor }(Y_{\mathbf{c},1}^{\left\lfloor e^{x}\right\rfloor })\cdot e^{-x}dx.
\end{eqnarray*}

By observing $H(Y_{\mathbf{c},1}^{N}|E_{\mathbf{c},1}^{N})\le N\Sigma(\mathbf{c})$,
we know the limits are finite as 
\begin{eqnarray*}
\lefteqn{\frac{1}{T}\cdot\int_{0}^{T}J_{\alpha\left\lfloor e^{x}\right\rfloor }^{\beta\left\lfloor e^{x}\right\rfloor }(Y_{\mathbf{c},1}^{\left\lfloor e^{x}\right\rfloor })\cdot e^{-x}dx}\\
 & \le & \frac{1}{T}\cdot\int_{0}^{T}\Sigma(\mathbf{c})\left\lfloor e^{x}\right\rfloor \cdot e^{-x}dx\\
 & \le & \Sigma(\mathbf{c}),
\end{eqnarray*}
and thus
\begin{equation}
\overline{J}_{\alpha}^{\beta}(Y_{\mathbf{c}},T)\le\Sigma(\mathbf{c}).\label{eq:J_upperbound}
\end{equation}

Similarly, by considering $J_{\alpha\left\lfloor e^{x}\right\rfloor }^{\beta\left\lfloor e^{x}\right\rfloor }(Y_{\mathbf{c},1}^{\left\lfloor e^{x}\right\rfloor })\le(\beta-\alpha)\left\lfloor e^{x}\right\rfloor $,
we can obtain
\begin{equation}
\overline{J}_{\alpha}^{\beta}(Y_{\mathbf{c}},T)\le\beta-\alpha.\label{eq:J_upperbound_seg}
\end{equation}

If $(\mathbf{c},r)$ is $\epsilon$-admissible, then whenever $r\ge\beta>\alpha\ge0$,
$M_{\left\lfloor N\alpha\right\rfloor +1}^{\left\lfloor N(\beta-\epsilon)\right\rfloor }$
can be decoded using $Y_{\mathbf{c},1}^{N}$ for sufficiently large
$N$. By Fano's inequality (note that the case where $\beta-\epsilon\le\alpha$
is obvious),
\begin{eqnarray*}
\lefteqn{\liminf_{T\to\infty}\frac{1}{T}\cdot\int_{0}^{T}J_{\alpha\left\lfloor e^{x}\right\rfloor }^{\beta\left\lfloor e^{x}\right\rfloor }(Y_{\mathbf{c},1}^{\left\lfloor e^{x}\right\rfloor })\cdot e^{-x}dx}\\
 & \ge & \liminf_{T\to\infty}\frac{1}{T}\cdot\int_{0}^{T}(\beta-\alpha-\epsilon)(1-H(\epsilon))\left\lfloor e^{x}\right\rfloor \cdot e^{-x}dx\\
 & = & (\beta-\alpha-\epsilon)(1-H(\epsilon)),
\end{eqnarray*}
and therefore
\begin{equation}
\liminf_{T\to\infty}\overline{J}_{\alpha}^{\beta}(Y_{\mathbf{c}},T)\ge(\beta-\alpha-\epsilon)(1-H(\epsilon)).\label{eq:J_fano}
\end{equation}

Also it is clear that for $\gamma<\alpha<\beta$,
\begin{equation}
\overline{J}_{\gamma}^{\beta}(Y_{\mathbf{c}},T)=\overline{J}_{\gamma}^{\alpha}(Y_{\mathbf{c}},T)+\overline{J}_{\alpha}^{\beta}(Y_{\mathbf{c}},T)\label{eq:seg_split}
\end{equation}

The following lemma can be readily observed.
\begin{lem}
\label{lem:subadd}For $\mathbf{c}=\mathbf{c}_{1}+\mathbf{c}_{2}$,
$\mathbf{c}_{1},\mathbf{c}_{2}\ge0$, we have
\[
\overline{J}_{\alpha}^{\beta}(Y_{\mathbf{c}_{1}},T)\le\overline{J}_{\alpha}^{\beta}(Y_{\mathbf{c}},T)\le\overline{J}_{\alpha}^{\beta}(Y_{\mathbf{c}_{1}},T)+\overline{J}_{\alpha}^{\infty}(Y_{\mathbf{c}_{2}},T).
\]
\end{lem}
\begin{proof}
As the quantities depend only on the marginal distributions of $\left\{ E_{\mathbf{c},i},X_{i},M_{i},Q\right\} $,
$\left\{ E_{\mathbf{c}_{1},i},X_{i},M_{i},Q\right\} $ and $\left\{ E_{\mathbf{c}_{2},i},X_{i},M_{i},Q\right\} $,
but not the joint distribution between $E_{\mathbf{c},i}$ and $E_{\mathbf{c}_{1},i}$
and so on. For the purpose of analysis, we assume the non-erasure
positions of Receiver $\mathbf{c}_{1}$ do not overlap with those
of Receiver $\mathbf{c}_{2}$, and erasure happens in Receiver $\mathbf{c}$
if an only if erasure happens in both Receiver $\mathbf{c}_{1}$ and
Receiver $\mathbf{c}_{2}$. Then $Y_{\mathbf{c},1}^{N}$ has the same
information as $\left(Y_{\mathbf{c}_{1},1}^{N},Y_{\mathbf{c}_{2},1}^{N}\right)$.
We can deduce from $I(X;Z)\le I(X,Y;Z)\le I(X;Z)+H(Y)$ that
\[
J_{N\alpha}^{N\beta}(Y_{\mathbf{c}_{1},1}^{N})\le J_{N\alpha}^{N\beta}(Y_{\mathbf{c},1}^{N})\le J_{N\alpha}^{N\beta}(Y_{\mathbf{c}_{1},1}^{N})+J_{N\alpha}^{\infty}(Y_{\mathbf{c}_{2},1}^{N}).
\]

The result follows.
\end{proof}
As the message is transmitted in a sequential manner, the received
symbols $Y_{\mathbf{c},1}^{N}$ should contain more information about
the older messages (the $M_{i}$ with smaller $i$) than newer messages,
and therefore the average information $(\beta-\alpha)^{-1}\cdot\overline{J}_{\alpha}^{\beta}(Y_{\mathbf{c}},T)$
should increase when $\alpha$ decrease. This property is proved in
the following lemma.
\begin{lem}
\label{lem:intervalext}Let $0\le\gamma\le\alpha<\beta$. We have
\[
\liminf_{T\to\infty}\left(\frac{1}{\beta-\gamma}\cdot\overline{J}_{\gamma}^{\beta}(Y_{\mathbf{c}},T)-\frac{1}{\beta-\alpha}\cdot\overline{J}_{\alpha}^{\beta}(Y_{\mathbf{c}},T)\right)\ge0.
\]
\end{lem}
\begin{proof}
Consider
\begin{eqnarray*}
J_{\gamma\left\lfloor e^{x}\right\rfloor }^{\beta\left\lfloor e^{x}\right\rfloor }(Y_{\mathbf{c},1}^{\left\lfloor e^{x}\right\rfloor }) & = & J_{\alpha\left\lfloor e^{x}\right\rfloor }^{\beta\left\lfloor e^{x}\right\rfloor }(Y_{\mathbf{c},1}^{\left\lfloor e^{x}\right\rfloor })+J_{\gamma\left\lfloor e^{x}\right\rfloor }^{\alpha\left\lfloor e^{x}\right\rfloor }(Y_{\mathbf{c},1}^{\left\lfloor e^{x}\right\rfloor })\\
 & \ge & J_{\alpha\left\lfloor e^{x}\right\rfloor }^{\beta\left\lfloor e^{x}\right\rfloor }(Y_{\mathbf{c},1}^{\left\lfloor e^{x}\right\rfloor })+J_{\gamma e^{x}}^{\alpha e^{x}}(Y_{\mathbf{c},1}^{\left\lfloor e^{x}\right\rfloor })-\alpha.
\end{eqnarray*}

After integrating the second term, we get, for any $x_{0}$,
\begin{eqnarray*}
\lefteqn{\int_{x_{0}}^{x_{0}+\log(\beta/\alpha)}J_{\gamma e^{x}}^{\alpha e^{x}}(Y_{\mathbf{c},1}^{\left\lfloor e^{x}\right\rfloor })dx}\\
 & \ge & \int_{x_{0}}^{x_{0}+\log(\beta/\alpha)}J_{\gamma e^{x}}^{\alpha e^{x}}(Y_{\mathbf{c},1}^{\left\lfloor e^{x_{0}}\right\rfloor })dx\\
 & = & \int_{x_{0}}^{x_{0}+\log(\beta/\alpha)}J_{\gamma e^{x}}^{\infty}(Y_{\mathbf{c},1}^{\left\lfloor e^{x_{0}}\right\rfloor })dx-\int_{x_{0}}^{x_{0}+\log(\beta/\alpha)}J_{\alpha e^{x}}^{\infty}(Y_{\mathbf{c},1}^{\left\lfloor e^{x_{0}}\right\rfloor })dx\\
 & = & \int_{x_{0}-\log(\alpha/\gamma)}^{x_{0}-\log(\alpha/\gamma)+\log(\beta/\alpha)}J_{\gamma e^{x+\log(\alpha/\gamma)}}^{\infty}(Y_{\mathbf{c},1}^{\left\lfloor e^{x_{0}}\right\rfloor })dx-\int_{x_{0}}^{x_{0}+\log(\beta/\alpha)}J_{\alpha e^{x}}^{\infty}(Y_{\mathbf{c},1}^{\left\lfloor e^{x_{0}}\right\rfloor })dx\\
 & = & \int_{x_{0}-\log(\alpha/\gamma)}^{x_{0}+\log(\beta\gamma/\alpha^{2})}J_{\alpha e^{x}}^{\infty}(Y_{\mathbf{c},1}^{\left\lfloor e^{x_{0}}\right\rfloor })dx-\int_{x_{0}}^{x_{0}+\log(\beta/\alpha)}J_{\alpha e^{x}}^{\infty}(Y_{\mathbf{c},1}^{\left\lfloor e^{x_{0}}\right\rfloor })dx\\
 & = & \int_{x_{0}-\log(\alpha/\gamma)}^{x_{0}}J_{\alpha e^{x}}^{\infty}(Y_{\mathbf{c},1}^{\left\lfloor e^{x_{0}}\right\rfloor })dx-\int_{x_{0}+\log(\beta\gamma/\alpha^{2})}^{x_{0}+\log(\beta/\alpha)}J_{\alpha e^{x}}^{\infty}(Y_{\mathbf{c},1}^{\left\lfloor e^{x_{0}}\right\rfloor })dx\\
 & = & \int_{x_{0}-\log(\alpha/\gamma)}^{x_{0}}J_{\alpha e^{x}}^{\infty}(Y_{\mathbf{c},1}^{\left\lfloor e^{x_{0}}\right\rfloor })dx-\int_{x_{0}-\log(\alpha/\gamma)}^{x_{0}}J_{\beta e^{x}}^{\infty}(Y_{\mathbf{c},1}^{\left\lfloor e^{x_{0}}\right\rfloor })dx\\
 & = & \int_{x_{0}-\log(\alpha/\gamma)}^{x_{0}}J_{\alpha e^{x}}^{\beta e^{x}}(Y_{\mathbf{c},1}^{\left\lfloor e^{x_{0}}\right\rfloor })dx\\
 & \ge & \int_{x_{0}-\log(\alpha/\gamma)}^{x_{0}}J_{\alpha e^{x}}^{\beta e^{x}}(Y_{\mathbf{c},1}^{\left\lfloor e^{x}\right\rfloor })dx\\
 & \ge & \int_{x_{0}-\log(\alpha/\gamma)}^{x_{0}}J_{\alpha\left\lfloor e^{x}\right\rfloor }^{\beta\left\lfloor e^{x}\right\rfloor }(Y_{\mathbf{c},1}^{\left\lfloor e^{x}\right\rfloor })dx-\alpha\log\frac{\alpha}{\gamma}.
\end{eqnarray*}

As a result,
\begin{eqnarray}
\lefteqn{\int_{x_{0}-\log(\beta/\alpha)}^{x_{0}}J_{\gamma\left\lfloor e^{x}\right\rfloor }^{\beta\left\lfloor e^{x}\right\rfloor }(Y_{\mathbf{c},1}^{\left\lfloor e^{x}\right\rfloor })dx}\nonumber \\
 & \ge & \int_{x_{0}-\log(\beta/\alpha)}^{x_{0}}J_{\alpha\left\lfloor e^{x}\right\rfloor }^{\beta\left\lfloor e^{x}\right\rfloor }(Y_{\mathbf{c},1}^{\left\lfloor e^{x}\right\rfloor })+\int_{x_{0}-\log(\beta/\alpha)}^{x_{0}}J_{\gamma e^{x}}^{\alpha e^{x}}(Y_{\mathbf{c},1}^{\left\lfloor e^{x}\right\rfloor })-\alpha\log\frac{\beta}{\alpha}\nonumber \\
 & \ge & \int_{x_{0}-\log(\beta/\alpha)}^{x_{0}}J_{\alpha\left\lfloor e^{x}\right\rfloor }^{\beta\left\lfloor e^{x}\right\rfloor }(Y_{\mathbf{c},1}^{\left\lfloor e^{x}\right\rfloor })+\int_{x_{0}-\log(\beta/\gamma)}^{x_{0}-\log(\beta/\alpha)}J_{\alpha\left\lfloor e^{x}\right\rfloor }^{\beta\left\lfloor e^{x}\right\rfloor }(Y_{\mathbf{c},1}^{\left\lfloor e^{x}\right\rfloor })dx-\alpha\log\frac{\beta}{\gamma}\nonumber \\
 & \ge & \int_{x_{0}-\log(\beta/\gamma)}^{x_{0}}J_{\alpha\left\lfloor e^{x}\right\rfloor }^{\beta\left\lfloor e^{x}\right\rfloor }(Y_{\mathbf{c},1}^{\left\lfloor e^{x}\right\rfloor })-\alpha\log\frac{\beta}{\gamma}.\label{eq:seg_int}
\end{eqnarray}

Fix any $T>0$. For each term in (\ref{eq:seg_int}), multiply it
with $e^{-x_{0}}$ and integrate it from 0 to $T$, we get 
\begin{eqnarray*}
\lefteqn{\int_{0}^{T}e^{-x_{0}}\int_{x_{0}-\log(\beta/\alpha)}^{x_{0}}J_{\gamma\left\lfloor e^{x}\right\rfloor }^{\beta\left\lfloor e^{x}\right\rfloor }(Y_{\mathbf{c},1}^{\left\lfloor e^{x}\right\rfloor })dxdx_{0}}\\
 & = & \int_{0}^{T}\int_{x}^{\min\left(T,x+\log(\beta/\alpha)\right)}e^{-x_{0}}J_{\gamma\left\lfloor e^{x}\right\rfloor }^{\beta\left\lfloor e^{x}\right\rfloor }(Y_{\mathbf{c},1}^{\left\lfloor e^{x}\right\rfloor })dx_{0}dx\\
 & \le & \int_{0}^{T}\int_{x}^{x+\log(\beta/\alpha)}e^{-x_{0}}J_{\gamma\left\lfloor e^{x}\right\rfloor }^{\beta\left\lfloor e^{x}\right\rfloor }(Y_{\mathbf{c},1}^{\left\lfloor e^{x}\right\rfloor })dx_{0}dx\\
 & = & \left(1-\frac{\alpha}{\beta}\right)\int_{0}^{T}e^{-x}J_{\gamma\left\lfloor e^{x}\right\rfloor }^{\beta\left\lfloor e^{x}\right\rfloor }(Y_{\mathbf{c},1}^{\left\lfloor e^{x}\right\rfloor })dx.
\end{eqnarray*}

And also
\begin{eqnarray*}
\lefteqn{\int_{0}^{T}e^{-x_{0}}\int_{x_{0}-\log(\beta/\gamma)}^{x_{0}}J_{\alpha\left\lfloor e^{x}\right\rfloor }^{\beta\left\lfloor e^{x}\right\rfloor }(Y_{\mathbf{c},1}^{\left\lfloor e^{x}\right\rfloor })dxdx_{0}}\\
 & = & \int_{0}^{T}\int_{x}^{\min\left(T,x+\log(\beta/\gamma)\right)}e^{-x_{0}}J_{\alpha\left\lfloor e^{x}\right\rfloor }^{\beta\left\lfloor e^{x}\right\rfloor }(Y_{\mathbf{c},1}^{\left\lfloor e^{x}\right\rfloor })dx_{0}dx\\
 & \ge & \int_{0}^{T-\log(\beta/\gamma)}\int_{x}^{x+\log(\beta/\gamma)}e^{-x_{0}}J_{\alpha\left\lfloor e^{x}\right\rfloor }^{\beta\left\lfloor e^{x}\right\rfloor }(Y_{\mathbf{c},1}^{\left\lfloor e^{x}\right\rfloor })dx_{0}dx\\
 & = & \left(1-\frac{\gamma}{\beta}\right)\int_{0}^{T-\log(\beta/\gamma)}e^{-x}J_{\alpha\left\lfloor e^{x}\right\rfloor }^{\beta\left\lfloor e^{x}\right\rfloor }(Y_{\mathbf{c},1}^{\left\lfloor e^{x}\right\rfloor })dx\\
 & \ge & \left(1-\frac{\gamma}{\beta}\right)\int_{0}^{T}e^{-x}J_{\alpha\left\lfloor e^{x}\right\rfloor }^{\beta\left\lfloor e^{x}\right\rfloor }(Y_{\mathbf{c},1}^{\left\lfloor e^{x}\right\rfloor })dx-\left(1-\frac{\gamma}{\beta}\right)(\beta-\alpha)\log\frac{\beta}{\gamma}.
\end{eqnarray*}

Therefore, by (\ref{eq:seg_int}),
\begin{eqnarray*}
\lefteqn{\left(1-\frac{\alpha}{\beta}\right)\int_{0}^{T}e^{-x}J_{\gamma\left\lfloor e^{x}\right\rfloor }^{\beta\left\lfloor e^{x}\right\rfloor }(Y_{\mathbf{c},1}^{\left\lfloor e^{x}\right\rfloor })dx}\\
 & \ge & \left(1-\frac{\gamma}{\beta}\right)\int_{0}^{T}e^{-x}J_{\alpha\left\lfloor e^{x}\right\rfloor }^{\beta\left\lfloor e^{x}\right\rfloor }(Y_{\mathbf{c},1}^{\left\lfloor e^{x}\right\rfloor })dx-\left(1-\frac{\gamma}{\beta}\right)(\beta-\alpha)\log\frac{\beta}{\gamma}-\int_{0}^{T}e^{-x_{0}}\alpha\log\frac{\beta}{\gamma}dx_{0}\\
 & = & \left(1-\frac{\gamma}{\beta}\right)\int_{0}^{T}e^{-x}J_{\alpha\left\lfloor e^{x}\right\rfloor }^{\beta\left\lfloor e^{x}\right\rfloor }(Y_{\mathbf{c},1}^{\left\lfloor e^{x}\right\rfloor })dx-\left(1-\frac{\gamma}{\beta}\right)(\beta-\alpha)\log\frac{\beta}{\gamma}-\left(1-e^{-T}\right)\alpha\log\frac{\beta}{\gamma}\\
 & \ge & \left(1-\frac{\gamma}{\beta}\right)\int_{0}^{T}e^{-x}J_{\alpha\left\lfloor e^{x}\right\rfloor }^{\beta\left\lfloor e^{x}\right\rfloor }(Y_{\mathbf{c},1}^{\left\lfloor e^{x}\right\rfloor })dx-\left(\frac{(\beta-\alpha)(\beta-\gamma)}{\beta}+\alpha\right)\cdot\log\frac{\beta}{\gamma}
\end{eqnarray*}

Multiply $\beta/\left(T\cdot(\beta-\alpha)(\beta-\gamma)\right)$
to both sides, we can obtain
\begin{eqnarray*}
\lefteqn{\frac{1}{\beta-\gamma}\cdot\frac{1}{T}\cdot\int_{0}^{T}e^{-x}J_{\gamma\left\lfloor e^{x}\right\rfloor }^{\beta\left\lfloor e^{x}\right\rfloor }(Y_{\mathbf{c},1}^{\left\lfloor e^{x}\right\rfloor })dx}\\
 & \ge & \frac{1}{\beta-\alpha}\cdot\frac{1}{T}\cdot\int_{0}^{T}e^{-x}J_{\alpha\left\lfloor e^{x}\right\rfloor }^{\beta\left\lfloor e^{x}\right\rfloor }(Y_{\mathbf{c},1}^{\left\lfloor e^{x}\right\rfloor })dx-\frac{1}{T}\cdot\left(1+\frac{\alpha\beta}{(\beta-\alpha)(\beta-\gamma)}\right)\cdot\log\frac{\beta}{\gamma}.
\end{eqnarray*}

Note that the second term vanishes when $T\to\infty$. The result
follows.
\end{proof}
We now proceed to prove an inequality on achievable rate-capacity
functions.
\begin{lem}
\label{lem:TwoSum}If a rate-capacity function $r(\mathbf{c})$ is
achievable, then for any $\mathbf{c}=\mathbf{c}_{1}+\mathbf{c}_{2}+...+\mathbf{c}_{n}$,
where $\mathbf{c}_{k}\ge0$ and $r(\mathbf{c})>r(\mathbf{c}_{k})$
for $k=1,...,n$, we have
\[
\sum_{k=1}^{n}\frac{\Sigma(\mathbf{c}_{k})-r(\mathbf{c}_{k})}{r(\mathbf{c})-r(\mathbf{c}_{k})}\ge1.
\]
\end{lem}
\begin{proof}
Without loss of generality, assume $r(\mathbf{c}_{1})\le...\le r(\mathbf{c}_{n})<r(\mathbf{c})$.
Let $\mathbf{c}_{0}=0$. Consider a code in which $r(\mathbf{c})$
is $\epsilon-$admissible. Fix any $k\in\{1,...,n\}$. By (\ref{eq:J_upperbound}),
for any $T\ge0$,

\[
\overline{J}_{0}^{\infty}(Y_{\mathbf{c}_{k}},T)\le\Sigma(\mathbf{c}_{k}).
\]

By (\ref{eq:seg_split}),
\begin{eqnarray*}
\overline{J}_{r(\mathbf{c}_{k})}^{\infty}(Y_{\mathbf{c}_{k}},T) & = & \overline{J}_{0}^{\infty}(Y_{\mathbf{c}_{k}},T)-\overline{J}_{0}^{r(\mathbf{c}_{k})}(Y_{\mathbf{c}_{k}},T)\\
 & \le & \Sigma(\mathbf{c}_{k})-\overline{J}_{0}^{r(\mathbf{c}_{k})}(Y_{\mathbf{c}_{k}},T).
\end{eqnarray*}

Invoking Lemma \ref{lem:subadd}, we obtain
\begin{eqnarray*}
\overline{J}_{r(\mathbf{c}_{k})}^{r(\mathbf{c})}(Y_{\mathbf{c}_{1}+...+\mathbf{c}_{k-1}},T) & \ge & \overline{J}_{r(\mathbf{c}_{k})}^{r(\mathbf{c})}(Y_{\mathbf{c}_{1}+...+\mathbf{c}_{k}},T)-\overline{J}_{r(\mathbf{c}_{k})}^{\infty}(Y_{\mathbf{c}_{k}},T)\\
 & \ge & \overline{J}_{r(\mathbf{c}_{k})}^{r(\mathbf{c})}(Y_{\mathbf{c}_{1}+...+\mathbf{c}_{k}},T)-\Sigma(\mathbf{c}_{k})+\overline{J}_{0}^{r(\mathbf{c}_{k})}(Y_{\mathbf{c}_{k}},T),
\end{eqnarray*}

It can be deduced using Lemma \ref{lem:intervalext} that
\begin{eqnarray*}
\liminf_{T\to\infty}\Biggl(\frac{1}{r(\mathbf{c})-r(\mathbf{c}_{k-1})}\overline{J}_{r(\mathbf{c}_{k-1})}^{r(\mathbf{c})}(Y_{\mathbf{c}_{1}+...+\mathbf{c}_{k-1}},T)\\
-\frac{1}{r(\mathbf{c})-r(\mathbf{c}_{k})}\overline{J}_{r(\mathbf{c}_{k})}^{r(\mathbf{c})}(Y_{\mathbf{c}_{1}+...+\mathbf{c}_{k-1}},T)\Biggr) & \ge & 0,
\end{eqnarray*}
and therefore
\begin{eqnarray*}
\liminf_{T\to\infty}\Biggl(\frac{1}{r(\mathbf{c})-r(\mathbf{c}_{k-1})}\overline{J}_{r(\mathbf{c}_{k-1})}^{r(\mathbf{c})}(Y_{\mathbf{c}_{1}+...+\mathbf{c}_{k-1}},T)\\
-\frac{1}{r(\mathbf{c})-r(\mathbf{c}_{k})}\overline{J}_{r(\mathbf{c}_{k})}^{r(\mathbf{c})}(Y_{\mathbf{c}_{1}+...+\mathbf{c}_{k}},T)+\frac{\Sigma(\mathbf{c}_{k})-\overline{J}_{0}^{r(\mathbf{c}_{k})}(Y_{\mathbf{c}_{k}},T)}{r(\mathbf{c})-r(\mathbf{c}_{k})}\Biggr) & \ge & 0,
\end{eqnarray*}

Summing through $k=1,...,n$,
\begin{eqnarray*}
\liminf_{T\to\infty}\Biggl(\frac{1}{r(\mathbf{c})-r(0)}\overline{J}_{r(0)}^{r(\mathbf{c})}(Y_{\mathbf{c}_{0}},T)\\
-\frac{1}{r(\mathbf{c})-r(\mathbf{c}_{n})}\overline{J}_{r(\mathbf{c}_{n})}^{r(\mathbf{c})}(Y_{\mathbf{c}},T)+\sum_{k=1}^{n}\frac{\Sigma(\mathbf{c}_{k})-\overline{J}_{0}^{r(\mathbf{c}_{k})}(Y_{\mathbf{c}_{k}},T)}{r(\mathbf{c})-r(\mathbf{c}_{k})}\Biggr) & \ge & 0,
\end{eqnarray*}
and thus
\begin{eqnarray*}
\liminf_{T\to\infty}\Biggl(\sum_{k=1}^{n}\frac{\Sigma(\mathbf{c}_{k})-\overline{J}_{0}^{r(\mathbf{c}_{k})}(Y_{\mathbf{c}_{k}},T)}{r(\mathbf{c})-r(\mathbf{c}_{k})}\\
-\frac{1}{r(\mathbf{c})-r(\mathbf{c}_{n})}\overline{J}_{r(\mathbf{c}_{n})}^{r(\mathbf{c})}(Y_{\mathbf{c}},T)\Biggr) & \ge & 0,
\end{eqnarray*}

Using (\ref{eq:J_fano}), we have 
\[
\liminf_{T\to\infty}\overline{J}_{0}^{r(\mathbf{c}_{k})}(Y_{\mathbf{c}_{k}},T)\ge\left(r(\mathbf{c}_{k})-\epsilon\right)\left(1-H(\epsilon)\right),
\]
 and 
\[
\liminf_{T\to\infty}\overline{J}_{r(\mathbf{c}_{n})}^{r(\mathbf{c})}(Y_{\mathbf{c}},T)\ge\left(r(\mathbf{c})-r(\mathbf{c}_{n})-\epsilon\right)\left(1-H(\epsilon)\right).
\]

Hence,
\begin{eqnarray*}
\sum_{k=1}^{n}\frac{\Sigma(\mathbf{c}_{k})-\left(r(\mathbf{c}_{k})-\epsilon\right)\left(1-H(\epsilon)\right)}{r(\mathbf{c})-r(\mathbf{c}_{k})}\\
-\frac{\left(r(\mathbf{c})-r(\mathbf{c}_{n})-\epsilon\right)\left(1-H(\epsilon)\right)}{r(\mathbf{c})-r(\mathbf{c}_{n})} & \ge & 0.
\end{eqnarray*}

Let $\epsilon\to0$. We obtained the desired result as
\[
\sum_{k=1}^{n}\frac{\Sigma(\mathbf{c}_{k})-r(\mathbf{c}_{k})}{r(\mathbf{c})-r(\mathbf{c}_{k})}\ge1.
\]

\end{proof}
We will discuss some cases in which superposition coding is provably
optimal in the next section.

\section{On-Off Multicast Networks}

In this section, we consider networks in which there are $d$ transmitters
and $2^{d}-1$ receivers, each having a different set of transmitters
to which it is connected. Transmitter $k$ broadcasts the same information
to the receivers it is connected to at rate $w_{k}$. Each receiver
has to decode the information at a different rate. There is no erasure
in the network. We would like to formulate the criteria on the decoding
rates of the receivers in which sequential data transmission is possible.

We first convert the problem into the multi-transmitter MRS setting.
Without loss of generality, we assume $w_{k}\le1$ for $k=1,...,d$.
We may replace a connection with rate $w$ by an erasure channel with
capacity $w$. Note that the decoding requirement does not depend
on the joint distribution of erasure events of different receivers.
Therefore, the problem can be translated to $d$-transmitter MRS.
We confine our study to the capacity vectors $\mathbf{c}\in\{0,w_{1}\}\times...\times\{0,w_{d}\}$.

\begin{figure}[h]
\centering{}\includegraphics[scale=1.2]{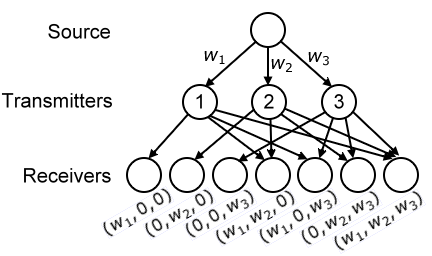}\caption{\label{Figure_GenNet}A three-transmitter on-off multicast network}
\end{figure}

\subsection{Non-optimality of superposition coding in general}

We can construct a network in which superposition coding is not optimal,
using a similar idea as in Example \ref{ex_NonOpt}.
\begin{example}
\label{ex_counter_3}The network contains 3 transmitters and 3 receivers,
where Transmitter 1 and 2 broadcast at rate 1, and Transmitter 3 broadcast
at rate 2. Receiver 1 is connected to Transmitter 1 and 2. Receiver
2 is connected to Transmitter 1 and 3. Receiver 3 is connected to
Transmitter 2 and 3. Non-superposition code can achieve the decoding
rate $3/2$ for Receiver 1, and $3$ for Receiver 2 and 3, which are
not achievable using superposition codes. Please refer to Subsection
\ref{sub:Non-superposition-codes} for the proof of achievability
and further discussions.
\end{example}
\begin{figure}[h]
\centering{}\includegraphics[scale=1.2]{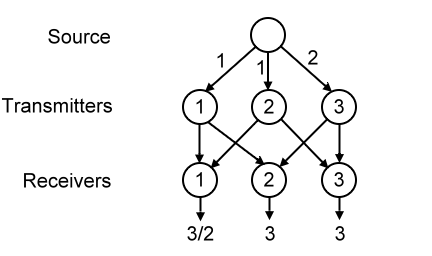}\caption{\label{Figure_Counter}The network in which superposition codes are
not optimal}
\end{figure}

Although superposition coding is not optimal for networks with 3 transmitters
in general, it is optimal for networks with 2 transmitters, which
will be shown in the following example.

\subsection{One-or-all on-off multicast network}
\begin{example}
[One-or-all on-off multicast network] There are $d$ transmitters
and $d+1$ receivers (numbered 0,...,$d$), where Transmitter $k$
broadcasts the same information to Receiver 0 and $k$ at rate of
$w_{k}$ bit/s for $k=1,...,d$. Receiver $k$ has to decode the data
at rate $r_{k}$ bit/s for $k=0,...,d$. We would like to characterize
the achievable region of $\left\{ r_{k}\right\} $ in terms of $\left\{ w_{k}\right\} $.
Figure \ref{Figure_2DAN} shows the network when $d=2$.
\end{example}
\begin{figure}[h]
\centering{}\includegraphics[scale=1.2]{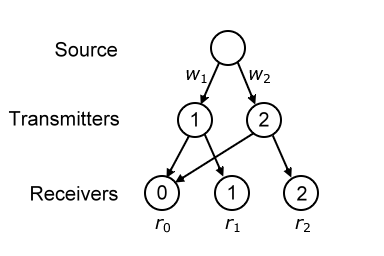}\caption{\label{Figure_2DAN}The two-transmitter one-or-all on-off multicast
network}
\end{figure}

The following theorem shows that superposition coding is optimal.
\begin{thm}
Superposition coding is optimal in the one-or-all on-off multicast
network, which has an achievable region
\[
r_{0}\ge0,\,0\le r_{k}\le w_{k}\text{ for }k=1,...,d,
\]
\[
\text{either }r_{0}\le\max(w_{k})\text{ or }\sum_{k=1}^{d}\frac{w_{k}-r_{k}}{r_{0}-r_{k}}\ge1.
\]
\end{thm}
\begin{proof}
The case where $r_{0}\le\max(w_{k})$ is trivial. We assume $r_{0}>w_{k}$
for all $k$. Without loss of generality, assume $r_{1}\le...\le r_{d}$,
then we have $r_{1}\le...\le r_{d}\le w_{d}<r_{0}$.

Among the inequalities in the proposed achievable region, $r_{k}\le w_{k}$
is obvious, and $\sum_{k=1}^{d}\frac{w_{k}-r_{k}}{r_{0}-r_{k}}\ge1$
is due to Lemma \ref{lem:TwoSum}. The converse follows.

We now show superposition MRS code can achieve the region. Assume
$\sum_{k=1}^{d}\frac{w_{k}-r_{k}}{r_{0}-r_{k}}\ge1$. Consider the
parameter
\[
g_{k}(\alpha)=\begin{cases}
\frac{1}{w_{k}} & \text{ when }\alpha\le r_{k}\\
\frac{w_{k}-r_{k}}{w_{k}(r_{0}-r_{k})} & \text{ when }r_{k}<\alpha\le r_{0}\\
0 & \text{ when }\alpha>r_{0}.
\end{cases}
\]

Note that
\[
\frac{w_{k}-r_{k}}{w_{k}(r_{0}-r_{k})}\le\frac{1}{w_{k}}
\]
due to $r_{0}\ge w_{k}$. Therefore $g(\alpha)$ is monotonically
decreasing along each dimension. It can be easily checked that $\int_{0}^{\infty}g_{k}(\alpha)d\alpha=1$.
It is left to check $\mathbf{c}\cdot g(r(\mathbf{c}))\ge1$.
\[
w_{k}\cdot g_{k}(r_{k})=w_{k}\cdot\frac{1}{w_{k}}=1,
\]
\begin{eqnarray*}
\lefteqn{\sum_{k=1}^{d}w_{k}\cdot g_{k}(r_{0})}\\
 & = & \sum_{k=1}^{d}w_{k}\cdot\frac{w_{k}-r_{k}}{w_{k}(r_{0}-r_{k})}\\
 & = & \sum_{k=1}^{d}\frac{w_{k}-r_{k}}{r_{0}-r_{k}}\ge1.
\end{eqnarray*}

By Theorem \ref{thm:MDsuperIff}, the region is achievable by superposition
MRS code.\end{proof}
\begin{rem*}
The achievable region in the MRS setting is different from that in
multilevel diversity coding, which admits the larger region
\[
r_{0}\ge0,\,0\le r_{k}\le w_{k}\text{ for }k=1,...,d,\, r_{0}+\sum_{k=1}^{d}r_{k}-\max_{k=1,...,d}(r_{k})\le\sum_{k=1}^{d}w_{k}.
\]

\end{rem*}

\subsection{Non-superposition codes\label{sub:Non-superposition-codes}}

In Example \ref{ex_counter_3}, we presented a network in which superposition
codes are non-optimal. Even though non-superposition codes are used,
the tools described in Section \ref{sec:Tools} may still be used.
We are going to prove the achievable region of a generalized version
of Example \ref{ex_counter_3}.
\begin{example}
\label{ex_3Dnonsuper}There are $3$ transmitters and $3$ receivers
numbered 1,2,$3$. Receiver 1 receives from Transmitter 1 and 2, and
wish to decode at $r_{1}$. Receiver 2 receives from Transmitter 1
and 3, and wish to decode at $r_{2}\ge r_{1}$. Receiver 3 receives
from Transmitter 2 and 3, and wish to decode at $r_{3}=r_{2}$.
\end{example}
\begin{figure}[h]
\centering{}\includegraphics[scale=1.2]{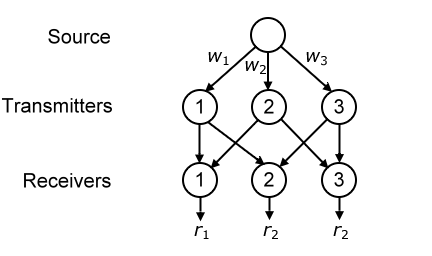}\caption{\label{Figure_3Dnonsuper}The network specified in Example \ref{ex_3Dnonsuper}}
\end{figure}

The achievable region is given by the following theorem.
\begin{thm}
The achievable region in Example \ref{ex_3Dnonsuper} can be given
by
\begin{eqnarray*}
w_{1}+w_{2} & \ge & r_{1}\\
w_{1}+w_{3} & \ge & r_{2}\\
w_{2}+w_{3} & \ge & r_{2}\\
w_{1}+w_{2}+w_{3}\cdot\frac{r_{2}-r_{1}}{r_{2}} & \ge & r_{2}.
\end{eqnarray*}
\end{thm}
\begin{proof}
[Proof of converse] The first three inequalities clearly hold. For
the last inequality, consider a code in which the rate requirements
are $\epsilon-$admissible. Receiver 1 can decode at rate $r_{1}$.
By (\ref{eq:J_fano}), 
\[
\underline{J}_{0}^{r_{1}}(Y_{(w_{1},w_{2},0)})\ge\left(r_{1}-\epsilon\right)\left(1-H(\epsilon)\right).
\]

Invoking (\ref{eq:J_upperbound}) and Lemma \ref{lem:subadd}, we
can obtain
\begin{eqnarray}
\underline{J}_{0}^{r_{1}}(Y_{(w_{1},0,0)})+w_{2} & \ge & \underline{J}_{0}^{r_{1}}(Y_{(w_{1},0,0)})+\overline{J}_{0}^{\infty}(Y_{(0,w_{2},0)})\nonumber \\
 & \ge & \underline{J}_{0}^{r_{1}}(Y_{(w_{1},w_{2},0)})\nonumber \\
 & \ge & \left(r_{1}-\epsilon\right)\left(1-H(\epsilon)\right).\label{eq:3x3l2_eq1}
\end{eqnarray}

Receiver 2 can decode at rate $r_{2}$. By (\ref{eq:J_fano}), 
\[
\underline{J}_{r_{1}}^{r_{2}}(Y_{(w_{1},0,w_{3})})\ge\left(r_{2}-r_{1}-\epsilon\right)\left(1-H(\epsilon)\right).
\]

Invoking Lemma \ref{lem:subadd}, we can obtain
\begin{eqnarray*}
\overline{J}_{r_{1}}^{\infty}(Y_{(w_{1},0,0)})+\underline{J}_{r_{1}}^{r_{2}}(Y_{(0,0,w_{3})}) & \ge & \underline{J}_{r_{1}}^{r_{2}}(Y_{(w_{1},0,w_{3})})\\
 & \ge & \left(r_{2}-r_{1}-\epsilon\right)\left(1-H(\epsilon)\right).
\end{eqnarray*}

By Lemma \ref{lem:intervalext} and (\ref{eq:J_upperbound}),
\[
w_{3}\ge\underline{J}_{0}^{r_{2}}(Y_{(0,0,w_{3})})\ge\frac{r_{2}}{r_{2}-r_{1}}\cdot\underline{J}_{r_{1}}^{r_{2}}(Y_{(0,0,w_{3})}),
\]
and therefore
\begin{equation}
\overline{J}_{r_{1}}^{\infty}(Y_{(w_{1},0,0)})+\frac{r_{2}-r_{1}}{r_{2}}\cdot w_{3}\ge\left(r_{2}-r_{1}-\epsilon\right)\left(1-H(\epsilon)\right).\label{eq:3x3l2_eq2}
\end{equation}

By (\ref{eq:J_upperbound}) and (\ref{eq:seg_split}),
\begin{eqnarray}
w_{1} & \ge & \overline{J}_{0}^{\infty}(Y_{(w_{1},0,0)})\nonumber \\
 & \ge & \underline{J}_{0}^{r_{1}}(Y_{(w_{1},0,0)})+\overline{J}_{r_{1}}^{\infty}(Y_{(w_{1},0,0)}).\label{eq:3x3l2_eq3}
\end{eqnarray}

Adding (\ref{eq:3x3l2_eq1}), (\ref{eq:3x3l2_eq2}) and (\ref{eq:3x3l2_eq3}),
\[
w_{1}+w_{2}+\frac{r_{2}-r_{1}}{r_{2}}\cdot w_{3}\ge\left(r_{2}-2\epsilon\right)\left(1-H(\epsilon)\right).
\]

The proof can be completed by taking $\epsilon\to0$.
\end{proof}
\[
\,
\]

\begin{proof}
[Proof of achievability] We will describe a coding scheme which
can achieve the proposed region. When $r_{1}=r_{2}$, we may use all
the transmitters to transmit at rate $r_{1}$. When $w_{3}\ge r_{2}$,
then we can use Transmitter 3 alone to transmit the message at rate
$r_{2}$, and Transmitter 1 together with Transmitter 2 to transmit
at rate $r_{1}$. Therefore we assume $r_{1}<r_{2}$ and $w_{3}<r_{2}$.

The code is specified by the block size $K$, the super-block size
$L>K$, and the parameter $\gamma$. Divide the message $\left\{ M_{i}\right\} _{i=1,2,...}$
into blocks of $K$ bits $B_{i}=M_{(i-1)K+1}^{iK}$. Further divide
each block $B_{i}$ into two sub-blocks $B_{1,i}$ and $B_{2,i}$,
where $B_{1,i}$ contains the first $\gamma K$ bits of the block,
where $\gamma$ is taken to be $w_{3}/r_{2}$ in this case, and $B_{2,i}$
contains the rest of the bits. Group the sub-blocks into super-blocks
by 
\[
S_{a,b,i}=\left(B_{b,(i-1)\cdot r_{a}\cdot L/K+1},...,B_{b,i\cdot r_{a}\cdot L/K}\right)
\]
for $a,b\in\{1,2\}$. Assume $r_{1}\cdot L/K$ and $r_{2}\cdot L/K$
are integers. Each super-block $S_{a,1,i}$ contains $r_{a}\cdot\gamma L$
bits, and each super-block $S_{a,2,i}$ contains $r_{a}\cdot(1-\gamma)L$
bits.

\begin{figure}[h]
\centering{}\includegraphics[scale=1.4]{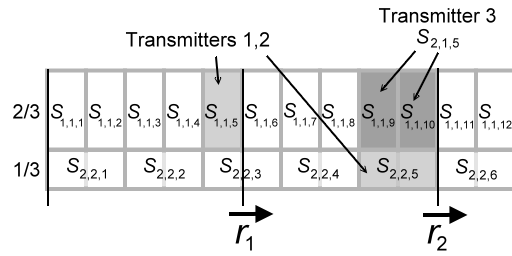}\caption{\label{Figure_nonsupercode}The super-blocks when $\gamma=2/3$ and
$r_{2}=2r_{1}$}
\end{figure}

At time $i$, Transmitter 1 encodes the super-blocks $S_{1,1,\left\lceil i/L\right\rceil }$
and $S_{2,2,\left\lceil i/L\right\rceil }$ using random linear projection
(concatenate the bits in the super-blocks and transmit a random projection
of the resultant vector). Transmitter 2 uses the same encoding scheme
as Transmitter 1. Transmitter 3 encodes the super-block $S_{2,1,\left\lceil i/L\right\rceil }$
using random linear projection.

To see why Receiver 2 can decode at rate $r_{2}$, assume that at
time $jL$, the super-blocks $S_{2,1,k}$ and $S_{2,2,k}$ are already
decoded for $k=1,...,j$. During the time interval $jL+1,...,(j+1)L$,
Receiver 2 will receive $w_{3}L$ bits from Transmitter 3, and $w_{1}L$
bits from Transmitter 1. For a bit encoded by Transmitter 1 at time
$i$, as the super-block $S_{1,1,\left\lceil i/L\right\rceil }$ is
already decoded (for $j$ large enough), it can be treated as a random
linear projection of $S_{2,2,\left\lceil i/L\right\rceil }$. Therefore
all together we have $w_{3}L$ projections of $S_{2,1,j+1}$, and
$w_{1}L$ projections of $S_{2,2,j+1}$. By definition of $\gamma$,

\[
w_{3}L=r_{2}\gamma L.
\]

From the assumption $w_{1}+w_{3}\ge r_{2}$,
\begin{eqnarray*}
w_{1}L & \ge & (r_{2}-w_{3})L\\
 & = & r_{2}(1-\gamma)L.
\end{eqnarray*}

Therefore both $S_{2,1,j+1}$ and $S_{2,2,j+1}$ can be decoded at
time $(j+1)K/r_{2}$. By induction, Receiver 2 can decode at rate
$r_{2}$. (Note that we may assume the first few blocks are decoded,
as we may allocate any extra amount of time to transmit them without
affecting the asymptotic behavior of the code). Similar result holds
for Receiver 3.

For Receiver 1, consider the time interval $jL+1,...,(j+1)L$. In
this time interval, Receiver 2 receives $w_{1}L$ bits from Transmitter
1, and $w_{2}L$ bits from Transmitter 2. All together there are $(w_{1}+w_{2})L$
projections of $S_{1,1,j+1}$ and $S_{2,2,j+1}$. The number of bits
in $S_{1,1,j+1}$ and $S_{2,2,j+1}$ is given by
\begin{eqnarray*}
r_{1}\cdot\gamma L+r_{2}\cdot(1-\gamma)L & = & \left(\frac{r_{1}w_{3}}{r_{2}}+r_{2}-w_{3}\right)\cdot L\\
 & = & \left(r_{2}-w_{3}\cdot\frac{r_{2}-r_{1}}{r_{2}}\right)\cdot L
\end{eqnarray*}
which is smaller than $(w_{1}+w_{2})L$ by the assumption $w_{1}+w_{2}+w_{3}\cdot\frac{r_{2}-r_{1}}{r_{2}}\ge r_{2}$.
Therefore Receiver 1 is able to decode at rate $r_{1}$.
\end{proof}
\[
\]

\begin{rem*}
Superposition codes can achieve a smaller region given by
\begin{eqnarray*}
w_{1}+w_{2} & \ge & r_{1}\\
w_{1}+w_{3} & \ge & r_{2}\\
w_{2}+w_{3} & \ge & r_{2}\\
w_{1}+w_{2}+2w_{3}\cdot\frac{r_{2}-r_{1}}{r_{2}} & \ge & 2r_{2}-r_{1}.
\end{eqnarray*}

\end{rem*}

\section{Conclusion and Discussions}

In this report, we have investigated the achievable regions and coding
schemes for multi-rate data transmission. We have shown that superposition
codes are optimal for the single transmitter setting. However, in
the multiple transmitter setting, there are some non-superposition
codes which outperform superposition codes.

Our results can be applied in various scenarios which requires the
messages to be decoded sequentially, for example, the broadcast streaming
of video. The multi-rate sequential data transmission setting can
also be applied on messages divided into several levels of importance.
For example, in the transmission of an interlaced image file, the
data corresponding to the low-resolution part is transmitted before
the data corresponding to the high-resolution part. A sequential code
can ensure that, even when the receiver has variable channel condition,
the low-resolution part is decoded first, and therefore the receiver
can display the image with lower resolution before all data are received.
If the connection might be stopped at any time, using a sequential
code can ensure the received message forms a continuous segment from
the beginning instead of fragmented data as in Fountain codes. In
the example of image transmission, if the connection is lost in the
middle of the transmission, the receiver can still decode a low-resolution
version of the image.

In Example \ref{ex_3Dnonsuper}, we have studied a particular 3-transmitter
network. Further investigation on the general 3-transmitter network,
and the next step, $N$-transmitter networks, may be carried out in
the future. Ultimately, we may consider general networks of interconnected
nodes instead of only two layers of nodes (transmitters and receivers).
Another direction is to find a method to construct non-superposition
codes according to the network connections and decoding rate requirements.
Example \ref{ex_3Dnonsuper} presents the construction of a non-superposition
code using a sub-block structure. It is left for future studies to
find out whether this construction method give the optimal code in
more general settings.

\bibliographystyle{IEEEtran}
\bibliography{ref}

\end{document}